\documentclass[12pt,a4paper]{article}%
\usepackage{adjustbox}
\usepackage{comment}
\usepackage{multicol}
\usepackage{multirow}
\usepackage{lmodern}
\usepackage[T1]{fontenc}
\usepackage[utf8]{inputenc}
\usepackage{parskip}
\usepackage{amsfonts}
\usepackage{xfrac}
\usepackage{booktabs}
\usepackage{pifont}
\usepackage{amsfonts,amssymb,amsthm,algorithm}
\usepackage{amsmath}
\usepackage{mathtools}
\usepackage{enumitem}
\usepackage[a4paper,margin=1in]{geometry}
\usepackage[figuresright]{rotating}
\usepackage[UKenglish]{isodate}
\usepackage{cases}
\DeclarePairedDelimiter\floor{\lfloor}{\rfloor}
\usepackage[authoryear]{natbib}
\usepackage[font=footnotesize,labelfont=sf]{caption}
\usepackage{setspace, tabularx}
\usepackage[page]{appendix} 
\usepackage{booktabs}
\usepackage[flushleft]{threeparttable}
\usepackage[table,dvipsnames]{xcolor}
\usepackage{accents}

\usepackage{array}
\newcolumntype{L}{>{$}l<{$}}
\usepackage{abstract}

\usepackage[colorlinks=true,
            urlcolor=blue,
            citecolor=blue,
            linkcolor=blue,
            bookmarks,
            bookmarksopen=true,
            bookmarksnumbered=true,
            pdfstartview={FitH}
]{hyperref}
 \newcommand{\Ex}{\textnormal{\textsf{E}}}

\newcommand{\cov}{\textnormal{\textsf{cov}}}

\newcommand{\eps}{\varepsilon}
\newcommand{\bI}{\textbf{1}}
\DeclareMathOperator*{\llog}{\textnormal{\textsf{log}}}

\DeclareMathOperator*{\argmax}{\textnormal{\textsf{arg\,max}}}
\DeclareMathOperator*{\llim}{\textnormal{\textsf{lim}}}
\DeclareMathOperator*{\llimsup}{\textnormal{\textsf{lim\,sup}}}
\DeclareMathOperator*{\ssup}{\textnormal{\textsf{sup}}}
\DeclareMathOperator*{\iinf}{\textnormal{\textsf{inf}}}
\DeclareMathOperator*{\mmax}{\textnormal{\textsf{max}}}

\newtheorem{definition}{Definition}

\newtheorem{assumption}{Assumption}
\newtheorem{lemma}{Lemma}
\newtheorem{proposition}{Proposition}

\usepackage[noblocks]{authblk}
\usepackage{graphicx}
\usepackage{tikz}
\usepackage{pgfplots}
    \usetikzlibrary{
        pgfplots.dateplot,
    }
\usepackage{subcaption}

\usepackage{amssymb}

\usepackage{graphicx,psfrag,epsf}
\usepackage{natbib} 
\usepackage{amsfonts,amssymb,amsthm,algorithm}
\usepackage{bm}    
\usepackage{threeparttable}   
\usepackage{tcolorbox}    
\usepackage{color,float}    
\usepackage{multirow}   
\usepackage{caption}
\usepackage{subcaption}
\usepackage{hyperref}
\usepackage{adjustbox}
\usepackage{setspace}
\def \d {\mathrm{d}}

\def\1{1\!{\rm l}}

\usepackage{pgfplots}
\pgfplotsset{compat=1.18}
\usepgfplotslibrary{groupplots}
\usepackage{pgfplotstable}
\usepackage[strings]{underscore} 

\begin{document}

\title{
Local Gaussian copula inference with structural breaks: testing dependence predictability
}

\author[a]{Alexander Mayer\thanks{Corresponding author. \emph{E-mail address}: \href{mailto:a.s.mayer@ese.eur.nl}{a.s.mayer@ese.eur.nl}. We thank Oliver Linton and seminar participants at the following seminars for valuable comments:   Universitat de les Illes Balears, Ca' Foscari Venice, and La Sapienza Rome.
}}
\author[b]{Tatsushi Oka}
\author[c]{Dominik Wied}
\affil[a]{\small {\it Erasmus University Rotterdam \& Tinbergen Institute, The Netherlands}}
\affil[b]{\small {\it Keio University, Japan}}
\affil[c]{\small {\it University of Cologne, Germany}}
\date{\today}
\maketitle 
\thispagestyle{empty}

\begin{abstract}
 
\vspace{0.2in}
We propose a score test for dependence predictability in conditional copulas that is robust to temporal instabilities. Our semiparametric procedure accommodates flexible dynamics in the marginal processes and remains agnostic about the copula family by leveraging distributional regression techniques together with a local Gaussian representation of the copula link function. We derive the limiting distribution of our test statistic and propose a resampling scheme based on recent results for the moving block bootstrap of multi-stage estimators.  Monte Carlo simulations and an empirical application illustrate the finite-sample performance of our methods.\\

\end{abstract}

\newgeometry{}
\setcounter{page}{1}

\section{Introduction}\label{sec:intro}

Understanding the dynamics of dependence between economic and financial time series is crucial for assessing systemic risk, market connectedness, and the role of common factors (e.g. \citealp{DieboldYilmaz2014Connectedness,  BrownleesEngle2017SRISK}). In many applications, researchers are not only interested in the marginal dynamics of time series but also in how their \emph{joint distribution} evolves over time and whether this evolution is affected by observable state variables (see, e.g., \citealp{patton:06,cappiello2006asymmetric,adrian:16}).

For example, in our empirical application to systemic risk, we ask whether lagged market downturns predict state-dependent quantile dependence between a financial institution's return and the contemporaneous market return, with particular interest in whether market downturns amplify joint downside risk (see also \citealp{han:2016}). A large empirical literature documents that dependence among equity returns strengthens during market stress and is often asymmetric, with different comovements in downturns and in upturns (see, e.g., \citealp{longin2001extreme,ang2002asymmetric,forbes2002no}). Several econometric contributions develop procedures to test whether dependence is constant over time or to detect structural changes in dependence (e.g., \citealp{wied:2012,buecher:14,stark:22,bomw:24}). Our focus is complementary: rather than characterizing dependence through global high/low regimes, we ask whether negative market states predict changes in dependence in particular parts of the joint distribution and during particular time periods.

For illustration, Figure~\ref{fig:ellipse} shows three scatter plots of the empirical ranks of standardized returns (upon filtering out conditional mean and variance dynamics) for major U.S.\ financial institutions, $r_t$, and the (CRSP value-weighted) market return, $r_{m,t}$, based on daily data from January 2000 to December 2019. Similar to \citet{han:2016}, we consider three financial institutions: JP Morgan ({\sf JPM}), Morgan Stanley ({\sf MS}), and {\sf AIG} that represent Depositories, Broker-Dealers, and Insurers, respectively. The plots are stratified by the sign of the lagged market return $r_{m,t-1}$, distinguishing days following a negative versus a positive market move. Comovement appears stronger following negative market returns \(r_{m,t-1}<0\), suggesting that market stress may increase dependence between \(r_t\) and \(r_{m,t}\). This is confirmed by a higher Spearman’s rho in the down market regime than in the up market regime and also (lower panel) quantile dependence (\citealp{patton2013copula}) appears to be higher following a bearish market. Put differently, Figure~\ref{fig:ellipse} points to two questions that correspond directly to our testing framework: ($i$) whether lagged market states have predictive content for dependence, and ($ii$) whether any such predictive effect is stable over time or instead subject to structural change. Motivated by this pattern, the goal of the present paper is to develop formal tools to test whether such state dependence is statistically significant and to locate \emph{where} in the joint distribution it occurs (e.g., in lower versus upper quantiles). In addition, we allow the dependence response to vary over time and investigate whether it is stable or exhibits structural change.

\begin{figure}
\centering
\begin{tikzpicture}
\begin{groupplot}[
  group style={group size=3 by 2, horizontal sep=.25cm, vertical sep=.8cm},
  width=5.0cm, height=5.0cm,
  xmin=0, xmax=1,
  ymin=0, ymax=1,
  clip=true,
  ticklabel style={font=\scriptsize},
  title style={font=\small},
]
 
\nextgroupplot[title={\sf MS}, axis equal image]
\addplot[only marks, mark=*, mark size=0.20pt,
  mark options={orange, opacity=0.2}]
  table[x=U1,y=U2,col sep=tab]{scatterMSZ0.txt};
\addplot[only marks, mark=*, mark size=0.20pt,
  mark options={blue, opacity=0.2}]
  table[x=U1,y=U2,col sep=tab]{scatterMSZ1.txt};
\addplot[thick, blue]   table[x=U1,y=U2,col sep=tab]{ellipseMSZ1.txt};
\addplot[thick, orange] table[x=U1,y=U2,col sep=tab]{ellipseMSZ0.txt};

\node[anchor=north west, font=\scriptsize] at (rel axis cs:0.02,0.98)
{$\rho^+=\text{0.453}\;\;\rho^-=\text{0.523}$};

\nextgroupplot[title={\sf AIG}, yticklabels={}, ytick=\empty, axis equal image]
\addplot[only marks, mark=*, mark size=0.20pt,
  mark options={orange, opacity=0.2}]
  table[x=U1,y=U2,col sep=tab]{scatterAIGZ0.txt};
\addplot[only marks, mark=*, mark size=0.20pt,
  mark options={blue, opacity=0.2}]
  table[x=U1,y=U2,col sep=tab]{scatterAIGZ1.txt};
\addplot[thick, orange] table[x=U1,y=U2,col sep=tab]{ellipseAIGZ0.txt};
\addplot[thick, blue]   table[x=U1,y=U2,col sep=tab]{ellipseAIGZ1.txt};

\node[anchor=north west, font=\scriptsize] at (rel axis cs:0.02,0.98)
{$\rho^+=\text{0.571}\;\;\rho^-=\text{0.613}$};

\nextgroupplot[title={\sf JPM}, yticklabels={}, ytick=\empty, axis equal image]
\addplot[only marks, mark=*, mark size=0.20pt,
  mark options={orange, opacity=0.2}]
  table[x=U1,y=U2,col sep=tab]{scatterJPMZ0.txt};
\addplot[only marks, mark=*, mark size=0.20pt,
  mark options={blue, opacity=0.2}]
  table[x=U1,y=U2,col sep=tab]{scatterJPMZ1.txt};
\addplot[thick, orange] table[x=U1,y=U2,col sep=tab]{ellipseJPMZ0.txt};
\addplot[thick, blue]   table[x=U1,y=U2,col sep=tab]{ellipseJPMZ1.txt};

\node[anchor=north west, font=\scriptsize] at (rel axis cs:0.02,0.98)
{$\rho^+=\text{0.672}\;\;\rho^-=\text{0.718}$};

\nextgroupplot[
  xlabel={},
  ylabel={},
  xmin=0.10, xmax=0.90,
  ymin=0.2, ymax=.8,
  title={},
  legend pos=south east,
  legend style={font=\scriptsize, fill=white, draw=none},
]

\addplot[thick, orange] table[x=q, y=z0, col sep=tab]{lowerMS.txt};
\addplot[thick, blue]   table[x=q, y=z1, col sep=tab]{lowerMS.txt};

\addplot[thick, orange, dashed] table[x=q, y=z0, col sep=tab]{upperMS.txt};
\addplot[thick, blue, dashed]   table[x=q, y=z1, col sep=tab]{upperMS.txt};

\nextgroupplot[yticklabels={}, ytick=\empty,   xmin=0.10, xmax=0.90,  ymin=0.2, ymax=.8,]
\addplot[thick, orange] table[x=q, y=z0, col sep=tab]{lowerAIG.txt};
\addplot[thick, blue]   table[x=q, y=z1, col sep=tab]{lowerAIG.txt};
\addplot[thick, orange, dashed] table[x=q, y=z0, col sep=tab]{upperAIG.txt};
\addplot[thick, blue, dashed]   table[x=q, y=z1, col sep=tab]{upperAIG.txt};

 \nextgroupplot[yticklabels={}, ytick=\empty,   xmin=0.10, xmax=0.90,  ymin=0.2, ymax=.8,]
\addplot[thick, orange] table[x=q, y=z0, col sep=tab]{lowerJPM.txt};
\addplot[thick, blue]   table[x=q, y=z1, col sep=tab]{lowerJPM.txt};
\addplot[thick, orange, dashed] table[x=q, y=z0, col sep=tab]{upperJPM.txt};
\addplot[thick, blue, dashed]   table[x=q, y=z1, col sep=tab]{upperJPM.txt};

\end{groupplot}
\end{tikzpicture}

\caption{The plots are based on daily CRSP return data from January 2000 to December 2019. Empirical ranks of standardized returns of three financial companies ($r_t$) and the market return ($r_{m,t}$) with 68\% Gaussian ellipses and Spearman’s rho superimposed (top row) and conditional quantile dependence curves (bottom row). Blue points/curves correspond to observations with $r_{m,t-1}<0$ (bear, $-$) while orange points/curves correspond to observations with $r_{m,t-1}>0$ (bull, $+$). In the bottom row, solid lines show lower-tail dependence $C(q,q)/q$, $q \in [0.1,0.5]$, and dashed lines show upper-tail dependence $(1-2q+C(q,q))/(1-q)$, $q \in [0.5,0.9]$ (e.g. \citealp[Eqs. (14)-(15)]{patton2013copula}). }
\label{fig:ellipse}
\end{figure}
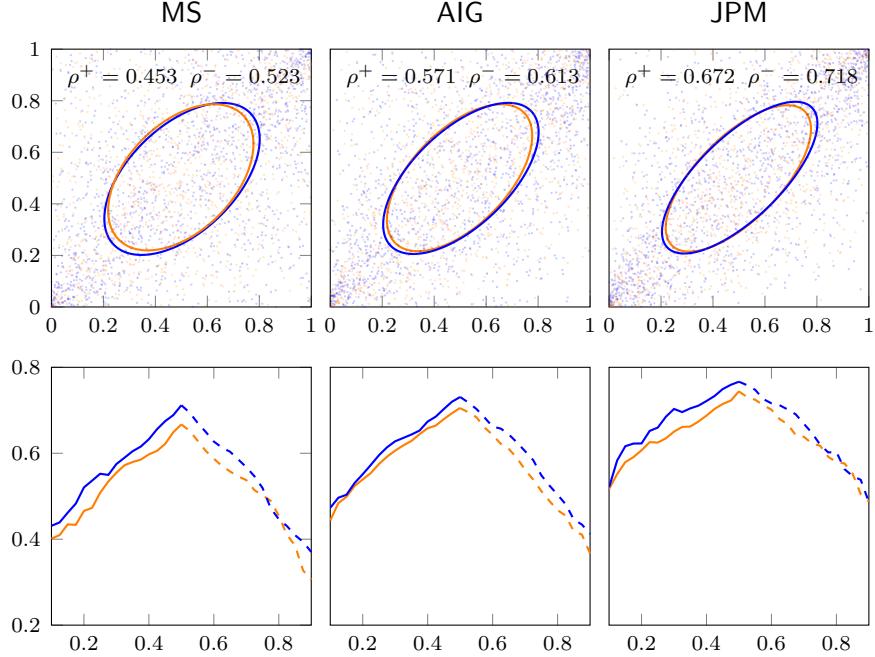

The problem studied in this paper is closely related to the Granger-causality literature, which has largely focused on the conditional mean or selected quantiles and therefore does not capture predictability in the full distributional dependence. This has motivated a growing literature on Granger causality in quantiles, including \cite{troster:18}, \citet{songetal:21}, and \citet{mwt:25}, as well as copula-based methods such as \citet{Bouezmarni:2012}, \citet{fuentes:2025}, \citet{lee:2014}, and \citet{jang:2024}. Moreover, there is a literature on how regressors influence copulas, including, among others, factor copulas \citep[e.g.][]{mw:2023}, vine copulas \citep[e.g.][]{jobst:2024}, and parametric copulas  \citep[e.g.][]{acar:13,gijbels:2017,gijbels:2021}.  We contribute to this line of work by testing whether a predetermined state vector carries predictive content for the conditional copula, after filtering the marginal dynamics. In this sense, our approach can be interpreted as a test for Granger-type predictability of dependence.

Our test is semiparametric and accommodates flexible AR-GARCH type dynamics in the marginal distributions, thereby allowing for empirically realistic features such as volatility clustering. Building on modern \emph{distribution regression} (DR) methods (see, e.g. \citealp{foresiperacchi:1995,cher:13,rw:13,wang:2023,wied:2024,spady:2025}), we employ a local Gaussian representation (see, e.g. \citealp{cherno24}), which enables us to remain fully agnostic about the true copula family. This strategy allows the dependence structure to be modelled locally at each individual point $(u,v)\in[0,1]^2$, rather than via global parametric restrictions (as, e.g., \citet{acar:13,gijbels:2017,gijbels:2021} do). 

A particular emphasis of our framework is on \emph{quantile dependence} (see, e.g., \citealp{patton:06,patton:12,patton2013copula,patton:13,oh2017modeling}). By examining the conditional copula along the main diagonal, we can study how co-movements, especially in the tails, vary with the state of the market. This perspective allows us, for example, to address whether the dependence between two return series changes with market conditions and whether such changes differ across quantiles and/or across time periods. Related ideas appear in, for instance, \citet{patton:06}, who examines time-varying quantile dependence in copula models. He shows that allowing copula parameters to vary over time is essential for capturing the asymmetric and state-dependent dependence observed in exchange rate returns, particularly the stronger lower-tail dependence during turbulent periods. His empirical results demonstrate that conditional copulas can reveal substantial time variation in co-movement that would remain hidden under static dependence models, motivating the need for flexible, state-dependent modelling frameworks like the one developed in this paper.

Our procedure is designed to be robust to structural breaks in the joint dependence structure in the spirit of \cite{sowell:96} and \citet{rossi:2005}. Using residuals from univariate AR-GARCH models removes marginal conditional mean/variance dynamics. A major contribution of this paper is thus a theory for time-series DR with estimated marginal processes. The cornerstone is a functional central limit theorem for a marked sequential empirical copula process with estimated pseudo-observations based on bracketing arguments from \citet{andrews1994introduction} and, in particular, the extensions to sequential processes by \citet{mohr2020weak} and \citet{scholze2024weak}. This result connects recent distributional regression ideas to the empirical copula process literature (e.g., \citealp{buecher:14, bucher2016dependent,  nasri2022change}) and provides the basis for our break-robust testing framework.  Finally, we suggest a moving block bootstrap suitable for multi-stage estimators \citep{gonccalves2023}.

The remainder of this paper is organized as follows: Section~\ref{sec:model} introduces the model and the testing problem. The test statistics and the bootstrap are described in Section~\ref{sec:teststat}, the asymptotic properties of which are derived in Section~\ref{sec:asy}.  Finite sample evidence in the form of a Monte Carlo experiment and an estimation exercise using financial data are presented in Sections~\ref{sec:MC} and  \ref{seq:emp}, respectively, before Section~\ref{sec:con} concludes. All proofs are delegated to the appendix.

\textbf{Notation.} For any random variable $X$ in $L^p$, let $\|X\|_p \coloneqq \Ex[|X|^p]^{1/p}$; for any vector $x = (x_1,\dots,x_m)^\top$, let $|x|=\sqrt{x^\top x}$ and $|x|_\infty = \mmax\{|x_1|,\dots,|x_m|\}$; for $u = (u_1,u_2) \in [0,1]^2$, let $u^{(1)} \coloneqq (u_1,1)$ and $u^{(2)} \coloneqq (1,u_2)$.

\section{Model and testing problem}\label{sec:model}

Consider the following location-scale specification
\begin{align}\label{eq:Ydgp}
Y_{j,t} = \mu_{j,t}(\lambda)+\sigma_{j,t}(\lambda)\eps_{j,t}, \qquad j \in \{1,2\},
\end{align}
where $\mu_{j,t}$ and $\sigma_{j,t}$ are known parametric functions up to $\lambda = \lambda_0$ that are ${\cal F}_{t-1}$-measurable. The innovations $\eps_{j,t}$ are {\it individually} independent of ${\cal F}_{t-1}$ for each $j \in \{1,2\}$, while $\eps_t = (\eps_{1,t},\eps_{2,t})^\top$ might depend {\it jointly}  on ${\cal F}_{t-1}$. This framework rules out time-varying parameters beyond the chosen parametric marginal model, but incorporates the important class of ARMA-GARCH models, similarly to what e.g. \citet{chen2006,chen2006estimation} or \citet{patton:13, oh2017modeling} consider. 

We assume that the marginal cdfs $Y_{j,t} \mid {\cal F}_{t-1} \sim {\sf H}_{j,t}$ are continuous. Following \cite{patton:06}, we can then decompose the conditional joint distribution  $Y_t  \mid {\cal F}_{t-1} \sim {\sf H}_t$, $Y_t \coloneqq (Y_{1,t},Y_{2,t})^\top$, in terms of a unique copula ${\sf C}_t$, i.e.
\[
{\sf H}_t(x_1,x_2) = {\sf C}_t({\sf H}_{1,t}(x_1),{\sf H}_{2,t}(x_2)).
\]
Due to the structure of Eq.~\eqref{eq:Ydgp}, it is known that we can express the preceding copula in terms of the innovation ranks:
\begin{align*}
    {\sf C}_t(u_1,u_2) = \, & {\mathbb P}\{U_{1,t} \leq u_1,U_{2,t} \leq u_2 \mid {\cal F}_{t-1}\},
\end{align*}
with
\begin{align*}
    U_{j,t} \coloneqq {\sf F}_j(\eps_{j,t}), \quad {\sf F}_j(x) = {\mathbb P}\{\eps_{j,t} \leq x\}, \quad j \in \{1,2\}.
\end{align*}

Moreover, upon filtering out the effect of ${\cal F}_{t-1}$ on $Y_t$,  assume furthermore that for a $k \times 1$  ${\cal F}_{t-1}$-measurable vector $Z_t \coloneqq (Z_{1,t},\dots,Z_{k,t})^\top$, the remaining dependence-relevant information in ${\cal F}_{t-1}$ is summarized by $Z_t$:

\begin{assumption}\label{ass:copZ}
    ${\mathbb P}\{U_{1,t} \leq u_1,U_{2,t} \leq u_2 \mid {\cal F}_{t-1}\}  = {\mathbb P}\{U_{1,t} \leq u_1,U_{2,t} \leq u_2 \mid Z_t\}$ for all $u = (u_1,u_2) \in [0,1]^2$.
\end{assumption}

Assumption~\ref{ass:copZ} is a maintained assumption that defines the information set with respect to which we assess predictability of the conditional copula. Although restrictive, Assumption~\ref{ass:copZ} is not uncommon in the literature and less restrictive than, for example, assuming that the copula is invariant to the conditioning information, i.e., ${\sf C}_t(\cdot\mid{\cal F}_{t-1})={\sf C}(\cdot)$ (see, e.g., \citealp{patton:13, oh2017modeling}). In contrast, following \cite{mw:2023} we allow dependence dynamics through a predetermined state vector $Z_t$. Extending our limit theory beyond such a state reduction would be technically non-trivial as discussed below.  Therefore, dependence non-predictability in copulas corresponds to whether the copula ${\sf C}_t(u \mid Z_t)$ varies with $Z_t$; i.e. we test whether $Z_t$ has predictive content for the conditional dependence between $Y_{1,t}$ and $Y_{2,t}$. In doing so, we also want to be robust to temporal instabilities. 

Our approach is based on the following (conditional) Gaussian representation (see \cite{kolev2006copulas} and \cite{cherno24} as well as the references therein): Let $\Phi_2(y_1,y_2)$ be the bivariate normal cdf. Then, for each $(u,z)$, there exists a function $\rho_t(u;z) \in [-1,1]$ such that
\begin{align}\label{eq:localgauss}
{\sf C}_t(u \mid z) = \Phi_2( \Phi^{-1}(u_1), \Phi^{-1}(u_2); \rho_t(u; z)) \eqqcolon C(u; \rho_t(u;z)).
\end{align}
Note that the local Gaussian representation is a pointwise reparametrization and thus not a parametric assumption. However, in order to make \eqref{eq:localgauss} operational, we impose a parametric structure on the local dependence parameter $\rho_t(u;z)$ using the Fisher transform:

\begin{assumption}\label{ass:rhoZ} Define  $$\varrho(u,\theta_t; z) \coloneqq {\sf tanh}( \alpha(u)+\beta_t(u)^\top z), \qquad \theta_t(u) \coloneqq (\alpha(u), \beta_t(u)^\top)^\top.$$ For each $u \in {\cal U} \subset (0,1)^2$ and $t \in {\mathbb N}$, there exists a $\theta_{0,t}(u)$ (potentially depending on $t$) such that Eq.~\eqref{eq:localgauss} holds with
\[
\rho_t(u;z) = \varrho(u,\theta_{0,t}; z).
\]
\end{assumption}

For a subset of interest ${\cal U} \subset (0,1)^2$, the null hypothesis $H_{0} = H_{0,1} \cap H_{0,2}$ can thus be rephrased as
\begin{align}\label{eq:H0}
H_{0,1}: \beta_{0,t}(u) = \beta_0(u)\; \forall (u,t) \in {\cal U} \times {\mathbb N} \quad H_{0,2}: \beta_{0}(u) = 0 \quad \forall u \in {\cal U}.
\end{align}
Under $H_0$, the restricted population parameter obtains as $\theta_0 = (\alpha_0(u),0)^\top$ so that, in view of Eq.~\eqref{eq:localgauss}, we get $$\varrho(u,\alpha_0) \coloneqq \varrho(u,\theta_0;z) ,\; C(u, \alpha) \coloneqq C(u; \varrho(u,\alpha)), \,\text{ and }\; C(u, \alpha_0) =  {\sf C}_t(u \mid z) = {\sf C}(u).$$ The testing problem is thus similar in nature to \cite{sowell:96} and \cite{rossi:2005} (see also \citealp{mwt:25}).

\section{Test statistic and bootstrap}\label{sec:teststat}

Since the population ranks $U_t = (U_{1,t},U_{2,t})^\top$ are usually unobserved, we will replace them with sequentially computed sample counterparts. In particular, define  as
$$\hat U_{j,t,s} \coloneqq \hat F_{j,s}(\hat \eps_{j,t}), \qquad j \in \{1,2\}, \quad s \in [0,1],$$
the ranks of the residuals $\hat\eps_{j,t} \coloneqq \eps_{j,t}(\hat\lambda)$, $\eps_{j,t}(\lambda) = (Y_{j,t}-\mu_{j,t}(\lambda))/\sigma_{j,t}(\lambda)$, calculated using the sequential empirical distribution function $\hat F_{j,s}(x) \coloneqq \hat F_{j,s}(x;\hat\lambda)$, with
\begin{align*}
\hat F_{j,s}(x;\lambda) \coloneqq \frac1{\floor{ns}}\sum_{t=1}^{\floor{sn}}\bI\{\eps_{j,t}(\lambda) \leq x\},
\end{align*}
using observations up to $\floor{s n}$, with $n$ denoting the length of the sample; the convention $\hat F_{j,s}(x;\lambda) = 0$ is used at $s = 0$. These pseudo observations $\hat U_{t,s} = (\hat U_{1,t,s},\hat U_{2,t,s})^\top$ involve two sources of first-step sampling uncertainty due to (1) estimation of $\lambda$ and (2) estimation of marginals ${\sf F}_j(\cdot)$ that will be taken into account when deriving the analytical properties of our test. 

Now, at $u = (u_1,u_2)^\top$, the log-likelihood is given by 
$$\hat\ell(u,s,\theta) \coloneqq \frac1{n}\sum_{t=1}^{\floor{sn}} \hat\ell_{t}(u,s,\theta;Z_t),$$ where
\begin{equation}\label{eq:likcont}
\begin{split}
\hat\ell_{t}(u,s,\theta;z) \coloneqq & \, \bI\{ \hat U_{1,t,s}\leq u_1, \hat U_{2,t,s}\leq u_2\}\llog C_t(u,\theta;z)  \\
\,& + \bI\{ \hat U_{1,t,s}\leq u_1, \hat U_{2,t,s}> u_2\}\llog (u_1 - C_t(u,\theta;z))\\
\,& + \bI\{ \hat U_{1,t,s}> u_1, \hat U_{2,t,s}\leq u_2\}\llog (u_2 -C_t(u,\theta;z))\\
\,& + \bI\{ \hat U_{1,t,s}> u_1, \hat U_{2,t,s}> u_2\}\llog (1-u_1-u_2 + C_t(u,\theta;z)),
    \end{split}
\end{equation}
with $C_t(u,\theta;z) \coloneqq C(u; \varrho_t(u,\theta; z))$. Note that, for a given $z$, $\hat\ell_{t}(u,s;\theta,z)$ is strictly concave in $C$, and, for a given $(u,z)$, $\theta \mapsto C(u,\theta;z)$ is monotone in $\theta$. It is convenient to represent the likelihood more succinctly via
\[
\hat\ell_{t}(u,s,\theta;z)= \hat d_{t}(u,s)^\top \llog p(u;\varrho_t(u,\theta;z)) =  \sum_{j=1}^4 \hat d_{j,t}(u,s) \llog p_j(u;\varrho_t(u,\theta;z)),
\]
with
 \begin{align*}
    \hat d_{t}(u,s) = ( \hat d_{1,t}(u,s),\dots,  \hat d_{4,t}(u,s))^\top \coloneqq 
    \begin{bmatrix}
    \bI\{\hat {U}_{1,t,s} \leq u_{1}\}\\ \bI\{\hat{U}_{1,t,s} > u_{1}\} \end{bmatrix} 
    \otimes 
    \begin{bmatrix} 
    \bI\{  \hat U_{2,t,s} \le u_{2}\}\\ \bI\{\hat{U}_{2,t,s} > u_{2}\} \end{bmatrix}
\end{align*}
denoting the $4\times 1$ vector of empirical events and 
\begin{align*}
p(u;\varrho_t(u,\theta;z))  =\,& (p_1(u;\varrho_t(u,\theta;z)),\dots,p_4(u;\varrho_t(u,\theta;z)))^\top \\
\coloneqq \,&
 \begin{bmatrix} 
    C_t(u,\theta; z) \\ 
    u_{1}-C_t(u,\theta; z)  \\ 
    u_{2}-C_t(u,\theta; z) \\
    1-u_1-u_2+C_t(u,\theta; z)
    \end{bmatrix}
\end{align*}
represent its population counterparts.

Our test will be based on the sequential score evaluated at the restricted ($\beta = 0$), full-sample estimator (defined below) $\hat\theta \coloneqq (\hat\alpha,0)^\top$:
\begin{align}
    \sqrt{n}\nabla_\theta\hat\ell(u,s;\hat\theta) = \tau(u,\hat\alpha)\frac1{\sqrt{n}}\sum_{t=1}^{\floor{s n}} X_t \sum_{j=1}^4e_j \frac{\hat d_{j,t}(u,s)}{p_j(u,\hat\alpha)}, \quad X_t \coloneqq (1,Z_t^\top)^\top,
\end{align}
where $(e_1,\dots,e_4)^\top \coloneqq (1,-1,-1,1)^\top$ and  $$\tau(u,\alpha) \coloneqq (1-\varrho^2(u,\alpha))C_\varrho(u,\alpha), \quad C_\varrho(u,\alpha) \coloneqq \frac{\partial}{\partial \varrho} C(u,\varrho) \Big\vert_{\varrho = \varrho(u,\alpha)}.$$ In order to ascertain the validity of the joint null in Eq.~\eqref{eq:H0}, we consider a test statistic that explicitly takes potential deviations from the two sub-hypotheses into account:
\begin{align}\label{eq:teststat}
T \coloneqq \omega(\Delta_1+\Delta_2), \quad \omega: \ell^\infty({\cal U}) \longrightarrow {\mathbb R}_+, \quad {\cal U} \subset (0,1)^2,
\end{align}
where $\omega$ is a suitable {\it aggregation} function (e.g. $\ssup_u \Delta(u)$ or $\llog \int \exp(\Delta(u)/2) \d u$, \citealp{andpol:1994}), and
$$
\Delta_1(u) \coloneqq \ssup\limits_{s \in [0,1]} \sqrt{n}|\nabla_\beta\hat\ell(u,s;\hat\theta)-s \nabla_\beta\hat\ell(u,1;\hat\theta)|_\infty, \quad \Delta_2(u) \coloneqq \sqrt{n}|\nabla_\beta\hat\ell(u,1;\hat\theta)|_\infty,
$$
are the CUSUM detector and the LM statistic, respectively, tailored to detect deviations from $H_{0,1}$ (parameter stability) and $H_{0,2}$ (constant non-Granger causality). Note that the CUSUM detector $\Delta_1$ does not require any trimming. Intuitively, $\Delta_1$ has no power against constant deviations from the null, and $\Delta_2$ lacks power if Granger causality is unstable; hence, combining both detectors ensures non-trivial power in both cases, an idea that goes back to \cite{sowell:96} and \cite{rossi:2005} (see also \citet{mwt:25}). 

As it turns out, the test statistic $T$ has an asymptotic distribution which depends on nuisance parameters (see Section~\ref{sec:asy}), so we propose to use a moving block bootstrap  approximation for the critical values. For a given bootstrap draw $b \in \{1,\dots,B\}$, one proceeds as follows:

\begin{enumerate}
  \item Generate moving block bootstrap indices
  $\tau^{(b)}=(\tau^{(b)}_1,\dots,\tau^{(b)}_n)$ with block length $l \coloneqq l(n)$ with $l(n) \rightarrow \infty$ such that $l(n) = o(n^{1/2})$, and form the bootstrap sample
    $(Y^{b}_{1,t},Y^{b}_{2,t},Z^{b}_t)=(Y_{1,\tau^{(b)}_t},Y_{2,\tau^{(b)}_t},Z_{\tau^{(b)}_t})$, $t=1,\dots,n$.
    \item As explained in \citet{gonccalves2023}, re-estimate the marginal models on
  $Y^{b}_{1,t}$ and $Y^{b}_{2,t}$ (including possibly $Z^{b}_t$), obtain standardized residuals and the sequential ranks.

  \item For each $u\in {\cal U}$, compute the restricted estimator
  $\hat\theta^{\,b}=(\hat\alpha^{\,b}(u),0)^\top$ by maximizing the restricted likelihood at $u$, construct
  $\nabla_\beta \hat\ell^{\,b}_n(u,s;\hat\theta^{\,b})$, $s\in[0,1]$, and compute
\begin{align*}
\Delta_1^b(u) = \,& \sqrt{n}\ssup\limits_{s\in[0,1]}|\nabla_\beta\ell_n^b(u,s,\hat\theta^b)-s\nabla_\beta\ell_n^b(u,1,\hat\theta^b)|_\infty \\
\Delta_2^b(u) = \,& \sqrt{n}|\nabla_\beta\ell_n^b(u,1,\hat\theta^b)-\nabla_\beta\ell_n(u,1,\hat\theta)|_\infty
\end{align*}
Set $T^{b}=\omega(\Delta^{\,b}_1+\Delta^{\,b}_2)$.
\item Use the empirical $(1-\alpha)$-quantile of $\{T^{b}\}_{b=1}^B$
as bootstrap critical value for $T$, and compute $p$-values analogously (and similarly for tests based on only
$\Delta_1$ or $\Delta_2$ if reported).
\end{enumerate}

Because, as discussed below, the first-step estimation error enters the $\Delta_2$ component through an additional drift term that does not cancel in $\Delta_2(u)= \sqrt{n}|\nabla_\beta \hat\ell(u,1;\hat\theta)|_\infty$, we center the bootstrap analogue by subtracting the original sample score, i.e., $\Delta_2^{\,b}(u)=\sqrt{n}|\nabla_\beta \ell_n^{\,b}(u,1;\hat\theta^{\,b})-\nabla_\beta \ell_n(u,1;\hat\theta)|_\infty$,
so that the bootstrap replicates the correct score fluctuations. On the other hand, if one is only interested in break testing, one can use an {\sf IID} bootstrap similar to \cite{nasri2022change}, because, as explained below, the limit of $\Delta_1$ is nuisance parameter free.

\section{Asymptotic properties}\label{sec:asy}

\subsection{Further assumptions}
In order to derive the properties of $T$, we impose the following assumptions:

\begin{assumption}\label{ass:Csmooth} If $H_0$ is true, then $\partial_i{\sf C}(u)$ exists and is continuous on $\{u \in [0,1]^2: u_i \in (0,1)\}$ for $i \in \{1,2\}$.
\end{assumption}

\begin{assumption}\label{ass:Z}
$Z_t$ is a $k \times 1$ strictly stationary process, ${\cal F}_{t-1}$-measurable, with $\| |Z_1| \|_p < \infty$ for some $p>2(d+4)$
\end{assumption}

\begin{assumption}\label{ass:marginal} $(1)$ $\sqrt{n}(\hat\lambda-\lambda_0) \Rightarrow \Lambda$, where $\Ex[\Lambda]=0$. Moreover, define
$$
a_{i,1,t}(\lambda) \coloneqq \frac{\mu_{i,t}(\lambda)-\mu_{i,t}(\lambda_0)}{\sigma_{i,t}(\lambda_0)}, \quad a_{i,2,t}(\lambda) \coloneqq \frac{\sigma_{i,t}(\lambda)}{\sigma_{i,t}(\lambda_0)}.
$$
\begin{itemize}
    \item[$(2)$] There exists an $r \times 1$ vector $R_t$ of strictly stationary and ${\cal F}_{t-1}$-measurable processes such that for each $t \geq 1$:
$$
\dot{a}_{i,j}(R_t) \coloneqq \nabla_\lambda a_{i,j,t}(\lambda) \large\mid_{\lambda = \lambda_0}, \quad 1 \leq i,j \leq 2,
$$
where $\| \dot{a}_{i,j}(R_1)\|_q < \infty$ for $q>2$ such that $q/(q-2) \leq p$, with $p$ from Assumption \ref{ass:Z} and $(Z_t^\top,R_t^\top)^\top$ is an $\alpha$-mixing strictly stationary process with mixing coefficient $\alpha(\tau) \leq ab^{\tau}$ for some $a \in (0,\infty)$, $b \in (0,1)$. 
\item[$(3)$] Let $H \subset {\mathbb R}^d$ be compact. Then, for each $t \geq 1$ and some $q > 2$
$$
\mmax\limits_{1 \leq i,j, \leq 2} \ssup\limits_{h \in H} |\nabla^2_\lambda a_{i,j,t}(\lambda_0+h/\sqrt{n})| \leq \ddot{a}_{t}, \qquad \Ex[\ddot{a}_{1}^q] <\infty.
$$    
\end{itemize}
\end{assumption}

\begin{assumption}\label{ass:epsdist} $\ssup\limits_{x \in \mathbb{R}}|f_j(x)| < \infty$, $\ssup\limits_{x \in \mathbb{R}}|xf_j(x)| < \infty$, $f_j({\sf F}_j^{-1}(x))(1+{\sf F}_j^{-1}(x)) = o(1)$ as $x \rightarrow 0$ or $x \rightarrow 1$. Moreover, for $1 \leq i,j \leq 2$:
  $$
  \ssup\limits_{(x_1,x_2) \in {\mathbb R}^2}|\partial_i\partial_j{\sf F}(x_1,x_2)(1+|x_i|)(1+|x_j|)|<\infty.
  $$
\end{assumption}

The smoothness condition is due to \cite{segers2012asymptotics} and is needed to apply the extended continuous mapping theorem.  Assumptions~\ref{ass:Z} and \ref{ass:marginal} are needed in order to apply \cite{mohr2020weak}, which is an extension of the empirical process CLT of \cite{andrews1994introduction} to the sequential case with unbounded classes of functions. Note that this involves a trade-off between moments of $Z$ and dependence allowed. In our empirical application $Z_t$ is an indicator and thus surely bounded. If one were interested only in the full-sample statistic ($s=1$) and $\beta$-mixing can be assumed, then $\Ex[|Z_1|^{4+2\delta}]$, $\delta > 0$ suffices.  This can be shown using \cite{doukhan1995invariance} as in \cite{neumeyer2019copula}. Assumption~\ref{ass:epsdist} is similar to conditions imposed by \cite{neumeyer2019copula} (see also \citealp{mw:2023}) and rules out marginal distributions with bounded support.

\subsection{Asymptotics of the restricted estimator}

Because our test statistic will be based on the score (with respect to $\beta$) of the {\it un}restricted DR problem  evaluated at the {\it restricted} estimator, we first need to understand the properties of the latter. Under the restriction $\beta(u) = 0$, we have $\theta(u) = (\alpha(u),0)^\top$ so that the full-sample ($s=1$) likelihood contribution
$$\hat\ell_{t}(u,1,\alpha,0;z) \eqqcolon \hat\ell_{t}(u,\alpha)$$ is independent of $z$. Define the restricted estimator that imposes $\beta(u) = 0$:
\begin{equation}\label{eq:restricted}
 \hat\alpha(u) \coloneqq   \argmax\limits_{\alpha \in {\cal A}} \hat\ell(u,\alpha), \quad   \hat\ell(u;\alpha) \coloneqq \frac1{n}\sum_{t=1}^n \hat\ell_{t}(u,\alpha).
 \end{equation}
Note that, by strict concavity of $\hat\ell(u,\alpha)$ in $C(u,\alpha) \coloneqq C(u, \varrho(u,\alpha))$ and strict monotonicity of $\alpha \mapsto C(u,\alpha)$,  the maximizer $\hat\alpha(u)$ is, for each $u \in {\cal U}$, unique. 

On the other hand, the population value $\alpha_0(u)$ can be viewed as the unique maximizer of the population objective function
\begin{align}\label{eq:ellinf}
    \ell(u,\alpha) \coloneqq  \Ex[\ell_t(u,\alpha)], \quad \ell_t(u,\alpha) = \sum_{j=1}^4 d_{j,t}(u) \llog p_j(u,\alpha),
\end{align} 
where, in contrast to $\hat\ell_t(u,\alpha)$ in Eq.~\eqref{eq:restricted}, the log-likelihood contribution $\ell_t(u,\alpha)$ involves the true $d_{j,t}(u)$ that are based on the population ranks $U_{j,t} \coloneqq {\sf F}_j(\eps_{j,t})$. If $H_0$ is true, then
${\mathbb P}\{d_{j,t}(u)=1\}=p_j(u,\alpha_0(u)) \eqqcolon p_j(u)$ so that the Kullback–Leibler divergence satisfies
\begin{align}
 \ell(u,\alpha_0(u))-\ell(u,\alpha) = \sum_{j=1}^4 p_j(u) \llog\frac{p_j(u)}{p_j(u,\alpha)} \geq 0,
\end{align}
 with equality if, and only if, $p_j(u,\alpha) = p_j(u)$. Since $\alpha(u) \mapsto C(u, \alpha)$ is strictly increasing 
 $$
 \alpha_0(u) \coloneqq \argmax\limits_{\alpha \in {\cal A}} \ell(u,\alpha)
 $$
is the unique maximizer for each $u \in {\cal U}$. Because $(u,\alpha) \mapsto \ell(u,\alpha)$ is continuous on a compact set ${\cal U} \times {\cal A}$ we obtain uniform separation:  
\begin{align}\label{eq:separation}    
\forall \eta>0:\;\iinf_{u \in {\cal U}} (\ell(u,\alpha_0(u)) - \iinf_{\alpha \in {\cal A}: |\alpha-\alpha_0(u)|>\eta} \ell(u,\alpha)) >0. 
\end{align}
Put differently, if $H_0$ is true, then $\alpha_0(\cdot)$ is uniformly identified on ${\cal U}$ and there is hope that it can be estimated using a sample of estimated ranks $(\hat U_1,\dots,\hat U_n)$. In order to show that this is indeed the case, we assume that the parameter space is compact: 

 \begin{assumption} ${\cal U} \times {\cal A} \subset (0,1)^2 \times (-\infty,\infty)$ is compact.
\end{assumption}

 \begin{proposition}\label{prop:alpha} Suppose $H_0$ is true.
 
 \begin{itemize}
     \item[$(1)$] $\ssup\limits_{u \in {\cal U}}|\hat\alpha(u)-\alpha_0(u)| = o_p(1).$     
\item[$(2)$]  If, in addition, $\alpha_0 \in {\sf int}({\cal A})$, then for $\tau(u) \coloneqq \tau(u,\alpha_0)$ we get
 \[
 \sqrt{n}(\hat\alpha-\alpha_0)(\cdot) \Rightarrow \frac{{\mathbb{C}}(\cdot)}{\tau(\cdot)} \quad \text{in} \;\; \ell^\infty({\cal U}),
 \]
with $\mathbb{C}(u)\coloneqq {\mathbb B}_C(u,1)- \sum_{k \in \{1,2\}} \partial_k{\sf C}(u){\mathbb B}_C(u^{(k)},1)$, where ${\mathbb B}_C$ is a Gaussian process with $\cov[{\mathbb B}_C(u,s),{\mathbb B}_C(v,t)] = (s \wedge t)\sigma_C(u,v)$, $ \sigma_C(u,v) \coloneqq {\sf C}(u \wedge v)-{\sf C}(u){\sf C}(v)$.
\end{itemize}
 \end{proposition}

The limiting distribution is governed by the same Gaussian process ${\mathbb C}$ that appears in \citet[Thm 3]{fermanian:2004} based on the Kiefer-Müller process ${\mathbb B}_C(u,1)$. Importantly, the well known (e.g. \citealp{gijbels2015estimation, remillard2017goodness}) invariance of the empirical copula process based on pseudo observations from location-scale models with respect to the first-step estimation error carries over to the DR regression framework considered here.

\subsection{Asymptotics of the test statistic}

To analyse the test statistic, it turns out convenient to define for any bivariate function $h: [0,1]^2 \rightarrow \mathbb{R}$ the real-valued functional 
\begin{equation}
\begin{split}\label{eq:varphiu}
\varphi_u(h(\cdot)) \coloneqq  \frac{h(u)}{p_1(u)}
+\frac{h(u)-h(u^{(1)})}{p_2(u)}
+&\frac{h(u)-h(u^{(2)})}{p_3(u)}\\
 +&\frac{h(u)-h(u^{(1)})-h(u^{(2)})+h(1,1)}{p_4(u)},
\end{split}
\end{equation}
where we recall that $p_j(u) = p_j(u,\alpha_0(u))$. Now, the crucial observation is that the $(k+1) \times 1$ gradient vector when evaluated at the true parameter vector $\theta_0 = (\alpha_0,0)^\top$ obeys:
\[
\sqrt{n}\nabla_\theta \hat\ell(u,s,\theta_0) = \sqrt{n}\begin{bmatrix}\partial_\alpha\hat\ell(u,s,\theta_0)  \\ \nabla_\beta\hat\ell(u,s,\theta_0) \end{bmatrix} =  \tau(u) \varphi_u(\mathbb{H}_n(\cdot,s;\hat\lambda)), 
\]
where $\varphi_u(\cdot)$, as defined in Eq.~\eqref{eq:varphiu}, is applied to each of the $k+1$ elements of the $(k+1) \times 1$ empirical process ${\mathbb H}_n$ given by
\begin{equation}
\begin{split}\label{eq:secproc}
      {\mathbb H}_n(u,s;\lambda) \coloneqq   \,& \sqrt{n}s_n(H_n(u,s;\lambda)- m{\sf C}(u)) \\
      H_n(u,s;\lambda) \coloneqq \,& \frac{1}{\floor{sn}}\sum_{t=1}^{\floor{sn}} X_t \bI\{\eps_{1,t}(\lambda)\leq \hat F_{1,s}^{-}(u_1;\lambda),\eps_{2,t}(\lambda)\leq \hat F_{2,s}^{-}(u_2;\lambda)\}, 
\end{split}
\end{equation}
where, by convention, $H_n(u,0;\lambda) \coloneqq 0$, and $X_t \coloneqq (1, Z_t^\top)^\top$, $m \coloneqq (1,\mu^\top)^\top$, with $\mu \coloneqq \Ex[Z_1]$.
 
Hence, it is clear that, in a first step, the properties of $\mathbb{H}_n$ need to be established to obtain the weak limit of the score. In a second step, the estimation error of $\hat\theta = (\hat\alpha,0)^\top$ from the previous section can be accounted for. 

The main idea is to relate the properties of the process ${\mathbb H}_n$ in Eq.~\eqref{eq:secproc}, based on the empirical cdf's, to its counterpart $\tilde {\mathbb H}_n$ that is based on the true cdf's:
\begin{equation}
\begin{split}\label{eq:tilde_secproc}
      \tilde {\mathbb H}_n(u,s;\lambda) \coloneqq   \,& \sqrt{n}s_n(\tilde H_n(u,s;\lambda)- m{\sf C}(u)) \\
      \tilde H_n(u,s;\lambda) \coloneqq \,& \frac{1}{\floor{sn}}\sum_{t=1}^{\floor{sn}} X_t \bI\{\eps_{1,t}(\lambda)\leq {\sf F}_{1}^{-}(u_1),\eps_{2,t}(\lambda)\leq {\sf F}_{2}^{-}(u_2)\},
        \end{split}    
\end{equation}
where we keep the convention $\tilde H_n(u,0;\lambda) \coloneqq 0$. In particular, upon decomposing $\tilde H_n = (\tilde C_n, \tilde G_n^\top)^\top$ according to the elements of $X_t = (1,\ Z_t^\top)^\top$, the following well-known identity (see, e.g., \citealp{segers2012asymptotics}) summarizes this idea:
\begin{align} \label{eq:gn_tilde_rel}
   H_n(u,s;\lambda) = \,& \tilde H_n\{\tilde C_{1,n}^{-}(u_1,s;\lambda), \tilde C_{2,n}^{-}(u_2,s;\lambda),s; \lambda\},   
\end{align}
where $\tilde C_{j,n}(u_j,s;\lambda) \coloneqq \tilde C_n(u^{(j)},s;\lambda)$, $j \in \{1,2\}$, are the marginal ecdf's corresponding to the empirical copula $\tilde C_n$. Thus, as argued elsewhere for the special case $Z_t = 1$ (see e.g. \citealp{bucher2013empirical} or \citealp{neumeyer2019copula}), we can think of Eq.~\eqref{eq:gn_tilde_rel} as a (two-argument) mapping that operates on a suitable function space and relates $H_n$ to $\tilde H_n$. The difference to the previous literature is that we also consider a marked version of the residual copula process, where the marks ($X_t$) depend on the regression residuals. For that reason, since the processes in Eq.~\eqref{eq:tilde_secproc} are not necessarily mean-zero even as $n$ diverges, we also define centred counterpart
\begin{align*}
\tilde {\mathbb H}^\circ_n(u,s;\lambda) \coloneqq \,& \sqrt{n }s_n(\tilde H_n(u,s;\lambda)-\Ex[\tilde H_n(u,s;\lambda)]).
\end{align*}
Note that $\tilde {\mathbb H}_n(\cdot,\cdot;\lambda_0)  = \tilde {\mathbb H}^\circ_n(\cdot,\cdot;\lambda_0)$.  

 As discussed below, provided the mapping underlying Eq.~\eqref{eq:gn_tilde_rel} is Hadamard-differentiability, the properties of $H_n$ can be linked to those of the functional input $\tilde H_n$. The preceding argument can be split into a stochastic part concerned with the weak convergence of the input processes and a deterministic part establishing the properties of the map.  The following lemma summarizes the stochastic part: 

\begin{lemma}\label{lem:seqtilde} Define $E \coloneqq [0,1]^2 \times [0,1] \times H$, where $H \subset \mathbb{R}^d$ is compact.
\begin{enumerate}
    \item[$(1)$] $\tilde {\mathbb H}^\circ_n(u,s;\lambda_0+h/\sqrt n)$ is stochastically equicontinuous in $(u,s,h) \in E$.
    \item[$(2)$] Define $\mathbb{B}(u,s) \coloneqq (\mathbb{B}_C(u,s),\mathbb{B}_Z(u,s)^\top)^\top$, where $\mathbb{B}_C \in \mathbb{R}$ $($defined in Proposition~\ref{prop:alpha}$)$ and $\mathbb{B}_Z \in \mathbb{R}^k$ are Gaussian processes with
    $$
    \cov[\mathbb{B}(u,s),\mathbb{B}(v,t)] = (s \wedge t)\Sigma(u,v), \quad \Sigma(u,v) \coloneqq  \begin{bmatrix} \sigma_C(u,v) & \sigma_{CZ}(u,v) \\
         \sigma_{CZ}(u,v)^\top & \sigma_Z(u,v)
    \end{bmatrix},
    $$
    where
    \begin{align*}
  \sigma_Z(u,v) =\Gamma_Z(0){\sf C}(u \wedge v)       + &{\sf C}(u){\sf C}(v)\sum_{j\geq 1}(\Gamma_Z(j)+\Gamma_Z(j)^\top) \\+& \sum_{j\geq 1}({\sf C}(v)\Delta_j(u)+{\sf C}(u)\Delta_j(v)^\top),
    \end{align*}  
    for $\normalfont\Delta_h(u) \coloneqq \Ex[(Z_1-\mu)(Z_{1+j}-\mu)(\bI\{U_1 \leq u\}-{\sf C}(u))]$, and $\normalfont\sigma_{CZ}(u,v) \coloneqq {\sf C}(u) \sum_{j \geq 1}\Ex[(Z_{1+j}-\mu)(\bI\{U_1\leq v\}-{\sf C}(v))].$    
    Then  
    $$\tilde{\mathbb H}_n(\cdot,\cdot;\lambda_0) = \tilde {\mathbb H}_n^\circ(\cdot,\cdot;\lambda_0)\Rightarrow {\mathbb H}(\cdot,\cdot) \coloneqq \begin{bmatrix}
      {\mathbb B}_C(\cdot,\cdot) \\
        \mu {\mathbb B}_C(\cdot,\cdot)+{\mathbb B}_Z(\cdot,\cdot) 
    \end{bmatrix},$$
 in $\ell^\infty([0,1]^2 \times [0,1]; {\mathbb R}^{k+1})$.
    \item[$(3)$] 
    $$\ssup\limits_{(u,s,h) \in E}|\tilde {\mathbb H}_n(u,s;\lambda_0+h/\sqrt{n})-\tilde  {\mathbb H}_n(u,s;\lambda_0) - s\Psi(u)h| = o_p(1),$$ where $\Psi \coloneqq (\psi_1(u),\dots,\psi_{k+1}(u))^\top \in \mathbb{R}^{(k+1) \times d}$, with
    $$
    \psi_j(u) \coloneqq  \sum_{p \in \{1,2\}}\partial_p {\sf C}(u)f_p({\sf F}_p^{-1}(u_k))\Ex[X_{1,j}\xi_p(u_p,R_1)], \quad j \in \{1,\dots,k+1\}
    $$
     for $\xi_p(u_p,r) \coloneqq {\sf F}_p^{-1}(u_p)\dot a_{p,2}(r)+\dot a_{p,1}(r)$, $p \in \{1,2\}.$
\end{enumerate}    
\end{lemma}

Lemma~\ref{lem:seqtilde} shows how the infeasible process $\tilde {\mathbb H}_n$ behaves uniformly in a local neighbourhood around the true parameter $\lambda_0$ that governs the marginal location scale models. Lemma~\ref{lem:seqtilde} holds on the whole $u \in [0,1]^2$ and not just on the compact subset ${\cal U} \subset (0,1)^2$. Moreover, if the null hypothesis is strengthened to independence between $\{Z_s\}_{s \geq 1}$ and $\{U_s\}_{s \geq 1}$, then the covariance kernel of $\Sigma(u,v)$ is block-diagonal and the limiting processes ${\mathbb B}_C$ and ${\mathbb B}_Z$ are independent.

Turning to the deterministic part, we show in the appendix that, akin to \citet[p. 64]{bucher2013empirical} and \citet[Appendix B]{bucher2016dependent}, the mapping defining Eq.~\eqref{eq:gn_tilde_rel} is Hadamard differentiable. The weak limit of ${\mathbb H}_n(\cdot,\cdot;\hat\lambda)$ follows then from the preceding lemma and the extended continuous mapping theorem. Since Lemma~\ref{lem:seqtilde} analyses $\tilde{\mathbb H}_n(\cdot,\cdot;\lambda_0+h/\sqrt n)$ for deterministic $h$, deriving the asymptotics of the feasible statistic requires replacing $h$ by the random drift $h_n=\sqrt n(\hat\lambda-\lambda_0)$.
To justify this substitution in the functional delta method, we impose the relatively weak assumption that Lemma~\ref{lem:seqtilde} (2) holds jointly with Assumption~\ref{ass:marginal} (1):

\begin{assumption}
    $(\tilde{\mathbb H}_n(\cdot,\cdot;\lambda_0),\sqrt{n}(\hat\lambda-\lambda_0)) \Rightarrow ({\mathbb H}, \Lambda)$
\end{assumption}

For a specific estimator of $\lambda$, the previous assumption can be derived under primitive assumptions.

\begin{lemma}\label{lem:sec_delt} 
${\mathbb H}_n(\cdot,\cdot;\hat\lambda) \Rightarrow {\mathbb H}(\cdot,\cdot) \coloneqq ({\mathbb C}(\cdot,\cdot),{\mathbb G}(\cdot,\cdot)^\top)^\top$, in  $\ell^\infty([0,1]^2 \times [0,1];{\mathbb R}^{k+1}),$ where ${\mathbb C}(u,s) = {\mathbb B}_C(u,s)-\sum_{j=1}^p\partial_j{\sf C}(u){\mathbb B}_C(u^{(j)},s)$ and ${\mathbb G}(u,s) = {\mathbb G}_0(u,s)+s\Xi(u)\Lambda,$
    $${\mathbb G}_0(u,s) = (\mu{\mathbb B}_C+{\mathbb B}_Z)(u,s)-\sum_{j=1}^2\partial_j{\sf C}(u)(\mu{\mathbb B}_C+{\mathbb B}_Z)(u^{(j)},s),$$
    with $\Xi(u) \coloneqq \sum_{j=1}^2\partial_j{\sf C}(u)f_j({\sf F}_j^{-1}(u_j))\cov[Z_1,\xi_p(u_j,R_1)].$
 Therefore, $$\sqrt{n}\nabla_\theta \hat\ell(\cdot,\cdot;\theta_0) \Rightarrow {\mathbb S}_{\theta_0}(\cdot,\cdot) \quad \text{ in } \ell^\infty({\cal U} \times [0,1];{\mathbb R}^{k+1}),$$ where ${\mathbb S}_{\theta_0}(u,s) = \tau(u)\varphi_u(\mathbb{H}(\cdot,s))$
\end{lemma}  

Lemma~\ref{lem:sec_delt} implies that the feasible marked sequential empirical process converges uniformly to a tight Gaussian process. The score with respect to $\alpha$ depends only on the true underlying copula and is
asymptotically unaffected by first-step estimation; in particular,
${\mathbb S}_{\alpha_0}(u,s)=\tau(u){\mathbb C}(u,s)\sum_{j=1}^4 1/p_j(u)$, where $p_j(u)$ are
the four quadrant probabilities. This is in line with earlier results on sequential empirical copula processes based on AR-GARCH residuals (see, e.g. \citealp{nasri2022change} or \citealp{bomw:24}).  In contrast, inference on $\beta$ generally inherits the additional drift term unless
$\cov[Z_1,\xi_p(u_p,R_1)]=0$, $p\in\{1,2\}$. Finally, we can use these results to obtain the limiting distribution of the test statistic:

\begin{proposition}\label{prop:seq_score_hat}
 $\sqrt{n}\nabla_\beta \hat\ell(\cdot,\cdot;\hat\theta) \Rightarrow {\mathbb S}_\beta$ in $\ell^\infty({\cal U} \times [0,1]; \mathbb{R}^k)$, where
 $$
  {\mathbb S}_\beta(u,s) \coloneqq\tau(u) [\varphi_u(\mathbb{G}(\cdot,s))-s\mu\varphi_u(\mathbb{C}(\cdot,1))] 
 $$
 and
 \[
(\Delta_1(\cdot),\Delta_2(\cdot)) \Rightarrow (\ssup\limits_{s \in [0,1]}|{\mathbb S}_\beta(\cdot,s)-s{\mathbb S}_\beta(\cdot,1)|_\infty, | {\mathbb S}_\beta(\cdot,1)|_\infty).
\]
\end{proposition}

Three points are worth noting: Firstly, because the Gaussian bridge ${\mathbb S}_\beta(\cdot,s)-s{\mathbb S}_\beta(\cdot,1)$ is independent of the endpoint ${\mathbb S}_\beta(\cdot,1)$, the limiting distribution of the statistics $\Delta_1$ and $\Delta_2$ are independent, too. Secondly, if $\mu = 0$, then $\sqrt{n}\nabla_\theta \hat\ell (\cdot,\cdot,\hat\theta) = \sqrt{n}\nabla_\theta \hat\ell (\cdot,\cdot,\theta_0) + o_p(1)$. Finally, $\Delta_1$ behaves as if you knew both $\lambda_0$ and $\alpha_0$; i.e.
\[
{\mathbb S}_\beta(u,s)-s{\mathbb S}_\beta(u,1) = \tau(u)(\varphi_u(\mathbb{G}_0(\cdot,s))-s\varphi_u(\mathbb{G}_0(\cdot,1))),
\]
using ${\mathbb G}(u,s) = {\mathbb G}_0(u,s)+s\Xi(u)\Lambda$. The estimation error does not, however, cancel in $\Delta_2$. For this reason we suggest the bootstrap scheme described in Section~\ref{sec:teststat}.

\section{Monte Carlo simulations}\label{sec:MC}

In this section, we study the finite-sample properties of the tests introduced in the previous sections using Monte Carlo experiments. In particular, we consider our test $T = \omega(\Delta_1+\Delta_2)$, as well as its individual components, i.e. the aggregated CUSUM test $\Delta_1 \coloneqq \omega(\Delta_1)$ and the aggregated LM test $\Delta_2 \coloneqq \omega(\Delta_2)$. Following \citet[Eq. 2.6]{andpol:1994}, the aggregation function $\omega: [0,1] \rightarrow {\mathbb R}_+$ is chosen as $\omega(f) = \llog \int_{\cal U}{\sf exp}\{\frac1{2}f(u)\}\d u$.

The state variable $Z_t$ is generated as a stationary AR(1) process,
$Z_t = 0.85 Z_{t-1} + v_t$, where the error term is ${\sf IID}$ with $v_t \sim \mathcal{N}(0,1-0.85^2)$.
Each marginal time series $Y_{j,t}$, $j\in\{1,2\}$, follows an AR(1)--GARCH(1,1) model:
\[
Y_{j,t} = \mu + \phi Y_{j,t-1} + \gamma Z_t + \eta_{j,t},\quad  
\sigma_{j,t}^2 = \omega + \alpha\eta_{j,t-1}^2 + \beta\,\sigma_{j,t-1}^2,
\]
where $\eta_{j,t}=\sigma_{j,t}\eps_{j,t}$ and $\eps_{j,t}$ are {\sf IID} standard Gaussian innovations. We set
$\mu=0$, $\phi=0.1$, $\gamma=1.5$, $\omega=0.01$, $\alpha=0.1$, and $\beta=0.85$. The AR(1)-GARCH(1,1)-parameters are estimated using the (quasi) Maximum Likelihood estimator from {\sf R}-package {\sf rugarch} (\citealp{rugarch}). After this, the ranks of the residuals are obtained in order to calculate the test statistics. 

We consider three copula specifications under the null:
($i$) a Gaussian copula with constant correlation, ($ii$) a Frank copula with constant parameter, and ($iii$) a Gaussian \emph{patchwork} copula \citep{Durante:2009} with correlation
$\rho(u)=\tanh(\alpha(u))$, where $\alpha(u)=1/4$ for $u\in[0,0.5]^2$ and $\alpha(u)=1/2$ otherwise. Intuitively, the copula in case ($iii$) is a rectangular patchwork that allows for different dependence in different subrectangles. Under the alternative, we consider a Gaussian patchwork copula with time-varying correlation
$\rho_t(u)={\sf tanh}(\alpha(u)+\beta_t(u)Z_t)$, where $\alpha(u)$ is as above.
For $u\in[0,0.5]^2$ we set: 
\begin{itemize}[leftmargin=\parindent, labelwidth=-1cm, labelsep=0.6em, align=left]
  \item[1. ({\sf constant})] $\beta_t(u)=1/2$ for all $t$;
  \item[2. ({\sf mid-break})] $\beta_t(u)=0$, $t\le\lfloor n/2\rfloor$ and $\beta_t(u)=1/2$,  $t>\lfloor n/2\rfloor$;
  \item[3. ({\sf offsetting break})] $\beta_t(u)=-1/2$, $t\le\lfloor n/2\rfloor$ and $\beta_t(u)=1/2$,  $t>\lfloor n/2\rfloor$.
\end{itemize}
We set $\beta_t(u)=0$ for $u\notin[0,0.5]^2$.

We focus on the quantile dependencies (e.g. \citealp{patton:12,patton2013copula} or \citealp{patton:13,oh2017modeling}), i.e. we restrict the analysis to the main diagonal $u=(u_1,u_2)^\top$ with $u_1=u_2=u$ and report results for the lower and upper diagonal regions
${\cal U}_l=[0.05,0.50]$ and ${\cal U}_u=[0.50,0.95]$. Following \citet{gonccalves2023}, the MBB block length $l$ is set equal to the next integer of the bandwidth from the automatic procedure for the Bartlett kernel of \cite{andrews1991heteroskedasticity}.

As $n$ increases, the empirical rejection rate approaches $5\%$ in all cases, the power properties reflect the design of the tests in all cases. In the case of the upper tail, none of the tests has non-trivial power as expected. In the case of the lower tail, the LM has highest power in the first alternative scenario (constant GC). In the mid-break scenario, the test still has power, but there is no non-trivial power in the offset scenario. The CUSUM test has highest power in the offset scenario, some power in the mid-break scenario and no power in the first alternative scenario. The combination of both test statistics always has non-trivial power, whereas the ranking is similar to the ranking of the LM test.

\begin{table}[!ht]
\centering
\small
\setlength{\tabcolsep}{3pt}
\renewcommand{\arraystretch}{1.08}

\begin{tabular}{llrccc c ccc}
\toprule
& & & \multicolumn{3}{c}{\textsf{size}} & & \multicolumn{3}{c}{\textsf{power}} \\
\cmidrule(lr){4-6}\cmidrule(lr){8-10}
 &  & $n$ & Gauss & Frank & Patch & & constant & mid-break & offset \\
\midrule

{\it lower} & $\Delta_1$
& 500     & 3.81\% & 3.21\% & 3.81\% & & 5.01\%  & 13.23\% & 41.88\% \\
& $\Delta_2$
&         & 3.01\% & 3.41\% & 3.81\% & & 40.68\% & 10.42\% & 4.61\% \\
& $T$
&         & 3.41\% & 3.81\% & 3.01\% & & 34.87\% & 12.22\% & 10.22\% \\
\cmidrule(lr){3-10}
& $\Delta_1$
& 1{,}000 & 3.41\% & 4.01\% & 3.81\% & & 5.81\%  & 22.04\% & 78.56\% \\
& $\Delta_2$
&         & 2.20\% & 3.81\% & 3.01\% & & 84.37\% & 23.25\% & 3.21\% \\
& $T$
&         & 2.00\% & 3.61\% & 3.61\% & & 80.56\% & 27.45\% & 16.83\% \\
\cmidrule(lr){3-10}
& $\Delta_1$
& 1{,}500 & 5.61\% & 6.41\% & 5.21\% & & 8.42\%  & 35.27\% & 92.99\% \\
& $\Delta_2$
&         & 3.61\% & 2.00\% & 3.21\% & & 96.79\% & 36.07\% & 4.81\% \\
& $T$
&         & 3.41\% & 2.61\% & 4.01\% & & 92.99\% & 45.49\% & 32.46\% \\
\midrule

{\it upper} & $\Delta_1$
& 500     & 4.81\% & 3.21\% & 2.61\% & & 1.60\% & 1.60\% & 1.40\% \\
& $\Delta_2$
&         & 2.20\% & 3.81\% & 1.60\% & & 3.41\% & 3.01\% & 2.61\% \\
& $T$
&         & 2.00\% & 4.21\% & 1.40\% & & 1.20\% & 1.60\% & 2.00\% \\
\cmidrule(lr){3-10}
& $\Delta_1$
& 1{,}000 & 4.81\% & 4.61\% & 4.41\% & & 4.41\% & 4.21\% & 4.61\% \\
& $\Delta_2$
&         & 4.81\% & 3.01\% & 4.01\% & & 4.21\% & 4.21\% & 4.21\% \\
& $T$
&         & 4.21\% & 3.41\% & 4.41\% & & 4.61\% & 4.81\% & 4.61\% \\
\cmidrule(lr){3-10}
& $\Delta_1$
& 1{,}500 & 4.61\% & 4.21\% & 5.81\% & & 5.21\% & 4.81\% & 4.81\% \\
& $\Delta_2$
&         & 3.81\% & 2.81\% & 2.61\% & & 3.81\% & 3.81\% & 4.21\% \\
& $T$
&         & 3.61\% & 3.21\% & 2.81\% & & 3.21\% & 3.01\% & 2.61\% \\

\bottomrule
\end{tabular}

\caption{Monte Carlo results (rejection frequencies in \%) at the 5\% nominal level based on 499 replications. 
Here, $\Delta_1$ denotes the CUSUM component, $\Delta_2$ the LM component, and $T=\omega(\Delta_1+\Delta_2)$ the combined test statistic. 
The MBB block length $l$ is set equal to the Andrews automatic bandwidth for the Bartlett kernel with $B=499$.}
\label{tab:MC}
\end{table}

\section{Empirical application}\label{seq:emp}

We now revisit the preliminary evidence from the introductory example, summarized in Figure~\ref{fig:ellipse} of Section~\ref{sec:intro}. Our goal is to examine how the dependence between the market return ($r_{m,t}$) and an individual asset return ($r_t$) varies with lagged market conditions. Specifically, let $r_{m,t}$ denote the daily CRSP value-weighted market index return. For $r_t$, we consider daily CRSP returns (Jan 2000-Dec 2019) of JP Morgan Chase ({\sf JPM}), Morgan Stanley ({\sf MS}), and {\sf AIG}, representing the Depositories, Broker-Dealers, and Insurance groups, respectively; see also \cite{han:2016}.

As discussed in Section~\ref{sec:model}, for each $u=(u_1,u_2)\in{\cal U}$, these co-movements are summarized by the local dependence measure $\varrho_t(u)$, which determines the conditional copula of $(r_t,r_{m,t})$. We distinguish between two predetermined market states $Z_t\in\{0,1\}$ corresponding to bear and bull market conditions. Specifically, we define the predetermined down market indicator
\(
Z_t \coloneqq \bI\{r_{m,t-1}<0\},
\)
so that $Z_t=1$ represents a \emph{bear} (down market) state and $Z_t=0$ a \emph{bull} (up market) state. Hence, for a given $u\in{\cal U}$ and in the absence of structural breaks, dependence is described by the state-specific local dependence measures
\[
\varrho^{\text{bear}}(u) = {\sf tanh}(\alpha(u)+\beta(u)) \quad \text{if } Z_t=1,
\qquad
\varrho^{\text{bull}}(u) = {\sf tanh}(\alpha(u)) \quad \text{if } Z_t=0.
\]

For both series, we fit GJR--AR(1)--GARCH(1,1) models \citep{glosten1993relation}. In the conditional mean of the individual asset return, we additionally include the market-state indicator $Z_t$. Estimation is carried out by quasi-maximum likelihood using the {\sf R}-package {\sf rugarch} (\citealp{rugarch}).

As motivated by the systemic-risk application, we focus on quantile dependence along the main diagonal of the conditional copula, i.e.,
$u=(u_1,u_2)^\top$ with $u_1=u_2$. As in the simulation study, we consider the lower and upper regions
\({\cal U}_{l}=[0.05,0.50],\) \({\cal U}_{u}=[0.50,0.95].
\)
Finally, the block length is set equal to the optimal bandwidth selected by the automatic procedure of \citet{andrews1991heteroskedasticity} for the Bartlett kernel.

If $T$ rejects, we apply the procedure of \citet{mwt:25} to determine whether the rejection is driven by $H_{1,1}$ or by $H_{1,2}$. We begin with a significance level of $\alpha=0.10$. If $T$ rejects at level $\alpha$, we compute the CUSUM detector $\Delta_1$. If the CUSUM test rejects at the level
\(
1-(1-\alpha)^{1/2}=0.051,
\)
we declare a break $s \in [0,1]$. We then set the adjusted significance level to
\(
1-(1-\alpha)^{1/4}=0.026
\)
and recompute $T$ on the subsamples $[0,s]$ and $[s,1]$. This Šidák correction ensures that the size corresponds to an overall significance level $\alpha$ (see also \citealp{galeano:2017}).

Overall, the results reported in Table~\ref{tab:emp} suggest that state-dependence becomes stronger during crisis periods. For {\sf MS}, we detect significant dependence predictability in both the lower and upper tails after 2005, but not before. For {\sf JPM}, we detect significant dependence predictability in the lower tail after 2002, but not before. For {\sf AIG}, we detect significant dependence predictability in the upper tail after 2002, but not before. Hence, across all three institutions, the evidence points to pronounced dependence predictability around 2008, often regarded as the climax of the global financial crisis associated with real-estate frictions (see \citealp{wied:2012}).

\begin{table}[ht]
\centering
\begin{tabular}{llrccc}
\toprule
&& & $[0,1]$ & $[0,s]$ & $[s,1]$ \\
\midrule

{\sf MS}&{\it lower}&$T$ & \textbf{0.000} & 0.144 & 0.028 \\
&&break   &{\tt 2005/10/11} &&\\[-5pt]
&&        &{\footnotesize (CU $p$=0.000)} &&\\
\cmidrule(lr){4-6}
&{\it upper}&$T$ & \textbf{0.000} & 0.317 & \textbf{0.008} \\
&&break   &{\tt 2005/10/11} &&\\[-5pt]
&&        &{\footnotesize (CU $p$=0.000)} &&\\
\midrule

{\sf JPM}&{\it lower}&$T$ &\textbf{ 0.002} & 0.385 & \textbf{0.000} \\
&&break   &{\tt 2002/07/02} &&\\[-5pt]
&&        &{\footnotesize (CU $p$=0.004)} &&\\
\cmidrule(lr){4-6}
&{\it upper}&$T$ & 0.246 &  &  \\
&&break   & -- &&\\[-5pt]
&&        & &&\\
\midrule

{\sf AIG}&{\it lower}&$T$ & 0.126 &  &  \\
&&break   & -- &&\\[-5pt]
&&        & &&\\
\cmidrule(lr){4-6}
&{\it upper}&$T$ & \textbf{0.024} & 0.994 & \textbf{0.012} \\
&&break   &{\tt 2002/06/14} &&\\[-5pt]
&&        &{\footnotesize (CU $p$=0.036)} &&\\

\bottomrule
\end{tabular}
\caption{Bootstrap (MBB) $p$-values using $B=499$ over 2000/01/01--2019/12/12. For the full-sample $[0,1]$ we use a first-step significance level $\alpha=0.10$; for the subsamples $[0,s]$ and $[s,1]$ (when a break is detected at $1-(1-\alpha)^{1/2} \simeq 0.051$) we use the adjusted level $1-(1-\alpha)^{1/4}\simeq 0.026$. The ``break'' row reports the estimated break date (if any), with the full-sample CUSUM $p$-value in parentheses.}
\label{tab:emp}
\end{table}

A possible economic interpretation is that crises amplify negative signals. This view is consistent with evidence that dependence and correlation tend to strengthen in downturns, in particular in the lower tail and in bear-market regimes (e.g., \citealp{longin2001extreme, ang2002asymmetric}). During crisis periods, a single negative market return may carry greater informational content and is less likely to be interpreted as a transitory fluctuation; instead, it is more readily viewed as confirmation of systemic stress. This interpretation is also in line with the systemic-risk literature, which emphasizes increases in tail co-movement under financial distress (\citealp{adrian:16}).

\section{Concluding remarks}\label{sec:con}

This paper proposes a flexible semiparametric test for dependence predictability that is designed to accommodate a broad class of marginal dynamics, remain agnostic about the copula family, and provide local information on how predetermined covariates affect the dependence structure in different parts of the distribution and over different periods. Although the primary application here is financial, related challenges are frequent in economics and beyond. For example, similar questions of Granger causality arise frequently in environmental econometrics (e.g. \citealp{Runge2019InferringCausation}).

Several technical extensions constitute promising directions for future research: First, it would be of interest to generalize the framework to multivariate settings with $J\ge 2$ variates $Y_{1,t},\dots,Y_{J,t}$, which would require addressing the increased complexity of the local Gaussian approximation. Secondly, in the spirit of \citet{patton:06}, one could let the state vector be constructed from lagged pseudo-observations, for example
\[
\hat Z_{t,s}(\lambda)\coloneqq \Phi^{-1}(\hat U_{1,t-1,s})\Phi^{-1}(\hat U_{2,t-1,s})
\]
thereby allowing for more flexible dependence dynamics driven by the history of the system. 
Thirdly, replacing the parametric marginal filtering step with nonparametric methods as in \citet{neumeyer2019copula} and \cite{chen2021efficient} could be an attractive alternative. Their results suggest that a $\sqrt{n}$-consistent DR estimation of $\theta$ might be possible. Finally, it may be of interest to relax Assumption~\ref{ass:copZ} by allowing for additional control variables $W_t$ when assessing the effect of $Z_t$, i.e., to consider conditional copulas of the form ${\sf C}(\cdot \mid {\cal F}_{t-1}) = {\sf C}(\cdot \mid W_t, Z_t)$.


\bibliography{bibl}

\begin{thebibliography}{65}
\newcommand{\enquote}[1]{``#1''}
\expandafter\ifx\csname natexlab\endcsname\relax\def\natexlab#1{#1}\fi

\bibitem[\protect\citeauthoryear{Acar, Craiu, and Yao}{Acar et~al.}{2013}]{acar:13}
\textsc{Acar, E.~F., R.~V. Craiu, and F.~Yao} (2013): \enquote{Statistical testing of covariate effects in conditional copula models,} \emph{Electronic Journal of Statistics}, 7, 2822--2850.

\bibitem[\protect\citeauthoryear{Adrian and Brunnermeier}{Adrian and Brunnermeier}{2016}]{adrian:16}
\textsc{Adrian, T. and M.~Brunnermeier} (2016): \enquote{CoVaR,} \emph{American Economic Review}, 106, 1705--1741.

\bibitem[\protect\citeauthoryear{Andrews}{Andrews}{1991}]{andrews1991heteroskedasticity}
\textsc{Andrews, D.~W.} (1991): \enquote{Heteroskedasticity and autocorrelation consistent covariance matrix estimation,} \emph{Econometrica: Journal of the Econometric Society}, 817--858.

\bibitem[\protect\citeauthoryear{Andrews and Pollard}{Andrews and Pollard}{1994}]{andrews1994introduction}
\textsc{Andrews, D.~W. and D.~Pollard} (1994): \enquote{An introduction to functional central limit theorems for dependent stochastic processes,} \emph{International Statistical Review/Revue Internationale de Statistique}, 119--132.

\bibitem[\protect\citeauthoryear{Andrews and Ploberger}{Andrews and Ploberger}{1994}]{andpol:1994}
\textsc{Andrews, D. W.~K. and W.~Ploberger} (1994): \enquote{Optimal Tests when a Nuisance Parameter is Present only under the Alternative,} \emph{Econometrica}, 62, 1383--1414.

\bibitem[\protect\citeauthoryear{Ang and Chen}{Ang and Chen}{2002}]{ang2002asymmetric}
\textsc{Ang, A. and J.~Chen} (2002): \enquote{Asymmetric correlations of equity portfolios,} \emph{Journal of Financial Economics}, 63, 443--494.

\bibitem[\protect\citeauthoryear{Borsch, Mayer, and Wied}{Borsch et~al.}{2024}]{bomw:24}
\textsc{Borsch, M., A.~Mayer, and D.~Wied} (2024): \enquote{Consistent Estimation of Multiple Breakpoints in Dependence Measures,} \emph{Journal of Business \& Economic Statistics}, 42, 695--706.

\bibitem[\protect\citeauthoryear{Bouezmarni, Rombouts, and Taamouti}{Bouezmarni et~al.}{2012}]{Bouezmarni:2012}
\textsc{Bouezmarni, T., J.~V. Rombouts, and A.~Taamouti} (2012): \enquote{Nonparametric Copula-Based Test for Conditional Independence with Applications to Granger Causality,} \emph{Journal of Business \& Economic Statistics}, 30, 275--287.

\bibitem[\protect\citeauthoryear{Brownlees and Engle}{Brownlees and Engle}{2017}]{BrownleesEngle2017SRISK}
\textsc{Brownlees, C.~T. and R.~F. Engle} (2017): \enquote{SRISK: A Conditional Capital Shortfall Measure of Systemic Risk,} \emph{The Review of Financial Studies}, 30, 48--79.

\bibitem[\protect\citeauthoryear{B{\"u}cher and Kojadinovic}{B{\"u}cher and Kojadinovic}{2016}]{bucher2016dependent}
\textsc{B{\"u}cher, A. and I.~Kojadinovic} (2016): \enquote{A dependent multiplier bootstrap for the sequential empirical copula process under strong mixing,} \emph{Bernoulli}, 2, 927--968.

\bibitem[\protect\citeauthoryear{B{\"u}cher and Volgushev}{B{\"u}cher and Volgushev}{2013}]{bucher2013empirical}
\textsc{B{\"u}cher, A. and S.~Volgushev} (2013): \enquote{Empirical and sequential empirical copula processes under serial dependence,} \emph{Journal of Multivariate Analysis}, 119, 61--70.

\bibitem[\protect\citeauthoryear{Bücher, Kojadinovic, Rohmer, and Segers}{Bücher et~al.}{2014}]{buecher:14}
\textsc{Bücher, A., I.~Kojadinovic, T.~Rohmer, and J.~Segers} (2014): \enquote{Detecting Changes in Cross-Sectional Dependence in Multivariate Time Series,} \emph{Journal of Multivariate Analysis}, 132, 111--128.

\bibitem[\protect\citeauthoryear{Cappiello, Engle, and Sheppard}{Cappiello et~al.}{2006}]{cappiello2006asymmetric}
\textsc{Cappiello, L., R.~F. Engle, and K.~Sheppard} (2006): \enquote{Asymmetric dynamics in the correlations of global equity and bond returns,} \emph{Journal of Financial Econometrics}, 4, 537--572.

\bibitem[\protect\citeauthoryear{Chen and Fan}{Chen and Fan}{2006{\natexlab{a}}}]{chen2006}
\textsc{Chen, X. and Y.~Fan} (2006{\natexlab{a}}): \enquote{Estimation and model selection of semiparametric copula-based multivariate dynamic models under copula misspecification,} \emph{Journal of Econometrics}, 135, 125--154.

\bibitem[\protect\citeauthoryear{Chen and Fan}{Chen and Fan}{2006{\natexlab{b}}}]{chen2006estimation}
---\hspace{-.1pt}---\hspace{-.1pt}--- (2006{\natexlab{b}}): \enquote{Estimation of copula-based semiparametric time series models,} \emph{Journal of Econometrics}, 130, 307--335.

\bibitem[\protect\citeauthoryear{Chen, Huang, and Yi}{Chen et~al.}{2021}]{chen2021efficient}
\textsc{Chen, X., Z.~Huang, and Y.~Yi} (2021): \enquote{Efficient estimation of multivariate semi-nonparametric GARCH filtered copula models,} \emph{Journal of Econometrics}, 222, 484--501.

\bibitem[\protect\citeauthoryear{Chernozhukov, Fern{\'a}ndez-Val, Han, and W{\"u}thrich}{Chernozhukov et~al.}{2024}]{cherno24}
\textsc{Chernozhukov, V., I.~Fern{\'a}ndez-Val, S.~Han, and K.~W{\"u}thrich} (2024): \enquote{Estimating Causal Effects of Discrete and Continuous Treatments with Binary Instruments,} \emph{arXiv:2403.05850}.

\bibitem[\protect\citeauthoryear{Chernozhukov, Fern{\'a}ndez-Val, and Melly}{Chernozhukov et~al.}{2013}]{cher:13}
\textsc{Chernozhukov, V., I.~Fern{\'a}ndez-Val, and B.~Melly} (2013): \enquote{Inference on Counterfactual Distributions,} \emph{Econometrica}, 81.

\bibitem[\protect\citeauthoryear{Diebold and Yilmaz}{Diebold and Yilmaz}{2014}]{DieboldYilmaz2014Connectedness}
\textsc{Diebold, F.~X. and K.~Yilmaz} (2014): \enquote{On the Network Topology of Variance Decompositions: Measuring the Connectedness of Financial Firms,} \emph{Journal of Econometrics}, 182, 119--134.

\bibitem[\protect\citeauthoryear{Doukhan, Massart, and Rio}{Doukhan et~al.}{1995}]{doukhan1995invariance}
\textsc{Doukhan, P., P.~Massart, and E.~Rio} (1995): \enquote{Invariance principles for absolutely regular empirical processes,} in \emph{Annales de l'IHP Probabilit{\'e}s et statistiques}, vol.~31, 393--427.

\bibitem[\protect\citeauthoryear{Durante, Saminger-Platz, and Sarkoci}{Durante et~al.}{2009}]{Durante:2009}
\textsc{Durante, F., S.~Saminger-Platz, and P.~Sarkoci} (2009): \enquote{Rectangular Patchwork for Bivariate Copulas and Tail Dependence,} \emph{Communications in Statistics - Theory and Methods}, 38, 2515--2527.

\bibitem[\protect\citeauthoryear{Fermanian, Radulovic, and Wegkamp}{Fermanian et~al.}{2004}]{fermanian:2004}
\textsc{Fermanian, J.-D., D.~Radulovic, and M.~Wegkamp} (2004): \enquote{Weak convergence of empirical copula processes,} \emph{Bernoulli}, 10, 847--860.

\bibitem[\protect\citeauthoryear{Forbes and Rigobon}{Forbes and Rigobon}{2002}]{forbes2002no}
\textsc{Forbes, K.~J. and R.~Rigobon} (2002): \enquote{No contagion, only interdependence: measuring stock market comovements,} \emph{The Journal of Finance}, 57, 2223--2261.

\bibitem[\protect\citeauthoryear{Foresi and Peracchi}{Foresi and Peracchi}{1995}]{foresiperacchi:1995}
\textsc{Foresi, S. and F.~Peracchi} (1995): \enquote{The Conditional Distribution of Excess Returns: An Empirical Analysis,} \emph{Journal of the American Statistical Association}, 90, 451--466.

\bibitem[\protect\citeauthoryear{Fuentes-Martínez, Crimaldi, and Rungi}{Fuentes-Martínez et~al.}{2025}]{fuentes:2025}
\textsc{Fuentes-Martínez, R., I.~Crimaldi, and A.~Rungi} (2025): \enquote{Non-linear dependence and Granger causality: A vine copula approach,} \emph{arXiv:2409.15070}.

\bibitem[\protect\citeauthoryear{Galanos}{Galanos}{2025}]{rugarch}
\textsc{Galanos, A.} (2025): \emph{rugarch: Univariate GARCH models}, package version 1.5-4.

\bibitem[\protect\citeauthoryear{Galeano and Wied}{Galeano and Wied}{2017}]{galeano:2017}
\textsc{Galeano, P. and D.~Wied} (2017): \enquote{Dating Multiple Change Points in the Correlation Matrix,} \emph{Test}, 26, 331--352.

\bibitem[\protect\citeauthoryear{Gijbels, Omelka, Pešta, and Veraverbeke}{Gijbels et~al.}{2017}]{gijbels:2017}
\textsc{Gijbels, I., M.~Omelka, M.~Pešta, and N.~Veraverbeke} (2017): \enquote{Score tests for covariate effects in conditional copulas,} \emph{Journal of Multivariate Analysis}, 159, 111--133.

\bibitem[\protect\citeauthoryear{Gijbels, Omelka, and Veraverbeke}{Gijbels et~al.}{2015}]{gijbels2015estimation}
\textsc{Gijbels, I., M.~Omelka, and N.~Veraverbeke} (2015): \enquote{Estimation of a copula when a covariate affects only marginal distributions,} \emph{Scandinavian Journal of Statistics}, 42, 1109--1126.

\bibitem[\protect\citeauthoryear{Gijbels, Omelka, and Veraverbeke}{Gijbels et~al.}{2021}]{gijbels:2021}
---\hspace{-.1pt}---\hspace{-.1pt}--- (2021): \enquote{Omnibus test for covariate effects in conditional copula models,} \emph{Journal of Multivariate Analysis}, 186, 104804.

\bibitem[\protect\citeauthoryear{Glosten, Jagannathan, and Runkle}{Glosten et~al.}{1993}]{glosten1993relation}
\textsc{Glosten, L.~R., R.~Jagannathan, and D.~E. Runkle} (1993): \enquote{On the relation between the expected value and the volatility of the nominal excess return on stocks,} \emph{The Journal of Finance}, 48, 1779--1801.

\bibitem[\protect\citeauthoryear{Gon{\c{c}}alves, Hounyo, Patton, and Sheppard}{Gon{\c{c}}alves et~al.}{2023}]{gonccalves2023}
\textsc{Gon{\c{c}}alves, S., U.~Hounyo, A.~J. Patton, and K.~Sheppard} (2023): \enquote{Bootstrapping two-stage quasi-maximum likelihood estimators of time series models,} \emph{Journal of Business \& Economic Statistics}, 41, 683--694.

\bibitem[\protect\citeauthoryear{Han, Linton, Oka, and Whang}{Han et~al.}{2016}]{han:2016}
\textsc{Han, H., O.~Linton, T.~Oka, and Y.-J. Whang} (2016): \enquote{The cross-quantilogram: Measuring quantile dependence and testing directional predictability between time series,} \emph{Journal of Econometrics}, 193, 251--270.

\bibitem[\protect\citeauthoryear{Jang, Kim, and Noh}{Jang et~al.}{2024}]{jang:2024}
\textsc{Jang, H., J.-M. Kim, and H.~Noh} (2024): \enquote{Vine copula Granger causality in quantiles,} \emph{Applied Economics}, 56, 1109--1118.

\bibitem[\protect\citeauthoryear{Jobst, Möller, and Groß}{Jobst et~al.}{2024}]{jobst:2024}
\textsc{Jobst, D., A.~Möller, and J.~Groß} (2024): \enquote{Gradient-Boosted Generalized Linear Models for Conditional Vine Copulas,} \emph{arXiv:2406.13500}.

\bibitem[\protect\citeauthoryear{Kolev, Anjos, and Mendes}{Kolev et~al.}{2006}]{kolev2006copulas}
\textsc{Kolev, N., U.~d. Anjos, and B.~V. d.~M. Mendes} (2006): \enquote{Copulas: a review and recent developments,} \emph{Stochastic models}, 22, 617--660.

\bibitem[\protect\citeauthoryear{Lee and Yang}{Lee and Yang}{2014}]{lee:2014}
\textsc{Lee, T.-H. and W.~Yang} (2014): \enquote{Granger-causality in quantiles between financial markets: Using copula approach,} \emph{International Review of Financial Analysis}, 33, 70--78.

\bibitem[\protect\citeauthoryear{Longin and Solnik}{Longin and Solnik}{2001}]{longin2001extreme}
\textsc{Longin, F. and B.~Solnik} (2001): \enquote{Extreme correlation of international equity markets,} \emph{The Journal of Finance}, 56, 649--676.

\bibitem[\protect\citeauthoryear{Mayer and Wied}{Mayer and Wied}{2023}]{mw:2023}
\textsc{Mayer, A. and D.~Wied} (2023): \enquote{Estimation and inference in factor copula models with exogenous covariates,} \emph{Journal of Econometrics}, 235, 1500--1521.

\bibitem[\protect\citeauthoryear{Mayer, Wied, and Troster}{Mayer et~al.}{2025}]{mwt:25}
\textsc{Mayer, A., D.~Wied, and V.~Troster} (2025): \enquote{Quantile Granger Causality in the Presence of Instability,} \emph{Journal of Econometrics}, 249(B), 105992.

\bibitem[\protect\citeauthoryear{Mohr}{Mohr}{2020}]{mohr2020weak}
\textsc{Mohr, M.} (2020): \enquote{A weak convergence result for sequential empirical processes under weak dependence,} \emph{Stochastics}, 92, 140--164.

\bibitem[\protect\citeauthoryear{Nasri, R{\'e}millard, and Bahraoui}{Nasri et~al.}{2022}]{nasri2022change}
\textsc{Nasri, B.~R., B.~N. R{\'e}millard, and T.~Bahraoui} (2022): \enquote{Change-point problems for multivariate time series using pseudo-observations,} \emph{Journal of Multivariate Analysis}, 187, 104857.

\bibitem[\protect\citeauthoryear{Neumeyer, Omelka, and Hudecov{\'a}}{Neumeyer et~al.}{2019}]{neumeyer2019copula}
\textsc{Neumeyer, N., M.~Omelka, and {\v{S}}.~Hudecov{\'a}} (2019): \enquote{A copula approach for dependence modeling in multivariate nonparametric time series,} \emph{Journal of Multivariate Analysis}, 171, 139--162.

\bibitem[\protect\citeauthoryear{Newey and McFadden}{Newey and McFadden}{1994}]{newmc:94}
\textsc{Newey, W.~K. and D.~McFadden} (1994): \enquote{Large sample estimation and hypothesis testing,} in \emph{Handbook of Econometrics}, Elsevier, vol.~4, chap.~36, 2111--2245.

\bibitem[\protect\citeauthoryear{Oh and Patton}{Oh and Patton}{2013}]{patton:13}
\textsc{Oh, D.~H. and A.~J. Patton} (2013): \enquote{Simulated Method of Moments Estimation for Copula-Based Multivariate Models,} \emph{Journal of the American Statistical Association}, 108, 689--700.

\bibitem[\protect\citeauthoryear{Oh and Patton}{Oh and Patton}{2017}]{oh2017modeling}
---\hspace{-.1pt}---\hspace{-.1pt}--- (2017): \enquote{Modeling dependence in high dimensions with factor copulas,} \emph{Journal of Business \& Economic Statistics}, 35, 139--154.

\bibitem[\protect\citeauthoryear{Patton}{Patton}{2013}]{patton2013copula}
\textsc{Patton, A.} (2013): \enquote{Copula methods for forecasting multivariate time series,} \emph{Handbook of Economic Forecasting}, 2, 899--960.

\bibitem[\protect\citeauthoryear{Patton}{Patton}{2006}]{patton:06}
\textsc{Patton, A.~J.} (2006): \enquote{Modelling asymmetric exchange rate dependence,} \emph{International Economic Review}, 47, 527--556.

\bibitem[\protect\citeauthoryear{Patton}{Patton}{2012}]{patton:12}
---\hspace{-.1pt}---\hspace{-.1pt}--- (2012): \enquote{A review of copula models for economic time series,} \emph{Journal of Multivariate Analysis}, 110, 4--18.

\bibitem[\protect\citeauthoryear{R{\'e}millard}{R{\'e}millard}{2017}]{remillard2017goodness}
\textsc{R{\'e}millard, B.} (2017): \enquote{Goodness-of-fit tests for copulas of multivariate time series,} \emph{Econometrics}, 5, 13.

\bibitem[\protect\citeauthoryear{Rio}{Rio}{2017}]{rio2017asymptotic}
\textsc{Rio, E.} (2017): \emph{Asymptotic theory of weakly dependent random processes}, vol.~80, Springer.

\bibitem[\protect\citeauthoryear{Rossi}{Rossi}{2005}]{rossi:2005}
\textsc{Rossi, B.} (2005): \enquote{Optimal Tests for Nested Model Selection with Underlying Parameter Instability,} \emph{Econometric Theory}, 21, 962--990.

\bibitem[\protect\citeauthoryear{Rothe and Wied}{Rothe and Wied}{2013}]{rw:13}
\textsc{Rothe, C. and D.~Wied} (2013): \enquote{Misspecification Testing in a Class of Conditional Distributional Models,} \emph{Journal of the American Statistical Association}, 108, 314--324.

\bibitem[\protect\citeauthoryear{Runge, Bathiany, Bollt, Camps-Valls, Coumou, Deyle, Glymour, Kretschmer, Mahecha, Muñoz-Marí, van Nes, Peters, Quax, Reichstein, Scheffer, Schölkopf, Spirtes, Sugihara, Sun, Zhang, and Zscheischler}{Runge et~al.}{2019}]{Runge2019InferringCausation}
\textsc{Runge, J., S.~Bathiany, E.~Bollt, G.~Camps-Valls, D.~Coumou, E.~Deyle, C.~Glymour, M.~Kretschmer, M.~D. Mahecha, J.~Muñoz-Marí, E.~H. van Nes, J.~Peters, R.~Quax, M.~Reichstein, M.~Scheffer, B.~Schölkopf, P.~Spirtes, G.~Sugihara, J.~Sun, K.~Zhang, and J.~Zscheischler} (2019): \enquote{Inferring causation from time series in Earth system sciences,} \emph{Nature Communications}, 10, 2553.

\bibitem[\protect\citeauthoryear{Scholze and Steland}{Scholze and Steland}{2024}]{scholze2024weak}
\textsc{Scholze, F.~A. and A.~Steland} (2024): \enquote{On the Weak Convergence of the Function-Indexed Sequential Empirical Process and its Smoothed Analogue under Nonstationarity,} \emph{arXiv preprint arXiv:2412.01635}.

\bibitem[\protect\citeauthoryear{Segers}{Segers}{2012}]{segers2012asymptotics}
\textsc{Segers, J.} (2012): \enquote{Asymptotics of empirical copula processes under non-restrictive smoothness assumptions,} \emph{Bernoulli}, 764--782.

\bibitem[\protect\citeauthoryear{Song and Taamouti}{Song and Taamouti}{2021}]{songetal:21}
\textsc{Song, X. and A.~Taamouti} (2021): \enquote{Measuring Granger Causality in Quantiles,} \emph{Journal of Business \& Economic Statistics}, 39, 937--952.

\bibitem[\protect\citeauthoryear{Sowell}{Sowell}{1996}]{sowell:96}
\textsc{Sowell, F.} (1996): \enquote{Optimal Tests for Parameter Instability in the Generalized Method of Moments Framework,} \emph{Econometrica}, 1085--1107.

\bibitem[\protect\citeauthoryear{Spady and Stouli}{Spady and Stouli}{2025}]{spady:2025}
\textsc{Spady, R. and S.~Stouli} (2025): \enquote{Gaussian Transforms Modeling and the Estimation of Distributional Regression Functions,} \emph{Econometrica}, forthcoming.

\bibitem[\protect\citeauthoryear{Stark and Otto}{Stark and Otto}{2022}]{stark:22}
\textsc{Stark, F. and S.~Otto} (2022): \enquote{Testing and Dating Structural Changes in Copula-based Dependence Measures,} \emph{Journal of Applied Statistics}, 49, 1121--1139.

\bibitem[\protect\citeauthoryear{Troster}{Troster}{2018}]{troster:18}
\textsc{Troster, V.} (2018): \enquote{Testing for Granger-Causality in Quantiles,} \emph{Econometric Reviews}, 37, 850--866.

\bibitem[\protect\citeauthoryear{van~der Vaart}{van~der Vaart}{1994}]{van1994weak}
\textsc{van~der Vaart, A.} (1994): \enquote{Weak convergence of smoothed empirical processes,} \emph{Scandinavian Journal of Statistics}, 501--504.

\bibitem[\protect\citeauthoryear{Wang, Oka, and Zhu}{Wang et~al.}{2023}]{wang:2023}
\textsc{Wang, Y., T.~Oka, and D.~Zhu} (2023): \enquote{Bivariate Distribution Regression with Application to Insurance Data,} \emph{Insurance: Mathematics and Economics}, 113, 215--232.

\bibitem[\protect\citeauthoryear{Wied}{Wied}{2024}]{wied:2024}
\textsc{Wied, D.} (2024): \enquote{Semiparametric Distribution Regression with Instruments and Monotonicity,} \emph{Labour Economics}, 90, 102565.

\bibitem[\protect\citeauthoryear{Wied, Kr{\"a}mer, and Dehling}{Wied et~al.}{2012}]{wied:2012}
\textsc{Wied, D., W.~Kr{\"a}mer, and H.~Dehling} (2012): \enquote{Testing for a change in correlation at an unknown point in time using an extended functional delta method,} \emph{Econometric Theory}, 28, 570--589.

\end{thebibliography}

\newpage\clearpage
\begin{appendices}

\renewcommand{\theequation}{\Alph{section}.\arabic{equation}}
\renewcommand{\thelemma}{\Alph{section}.\arabic{lemma}}
\renewcommand{\theremark}{\Alph{section}.\arabic{remark}}
\setcounter{equation}{0}
\setcounter{lemma}{0}

\section{Proof of Lemma 1} 

 Decompose the $(k+1) \times 1$ process $\tilde{\mathbb H}^\circ_n = (\tilde{\mathbb C}^\circ_n, \tilde{\mathbb G}_n^{\circ \top})^\top$ according to $X_t = (1,\ Z_t^\top)^\top$.
For simplicity, we provide the proof  for $\tilde{\mathbb G}^\circ_n$; the conclusions for $\tilde{\mathbb C}^\circ_n$ follow as a special case with $Z = 1$.

\subsection{Proof of Lemma 1 (1)} 
Recall the definition from Assumption~\ref{ass:marginal}:
$$
a_{i,1}({\cal F}_{t-1},\lambda) \coloneqq \frac{\mu_i({\cal F}_{t-1},\lambda)-\mu_i({\cal F}_{t-1},\lambda_0)}{\sigma_i({\cal F}_{t-1},\lambda_0)}, \quad a_{i,2}({\cal F}_{t-1},\lambda) \coloneqq \frac{\sigma_i({\cal F}_{t-1},\lambda)}{\sigma_i({\cal F}_{t-1},\lambda_0)}
$$
for each $i \in \{1,2\}$. It follows that $a_{i,1}({\cal F}_{t-1},\lambda_0) = 0$ and $a_{i,2}({\cal F}_{t-1},\lambda_0) = 1$. Let 
$$T_{i,t}(u_i,\lambda) \coloneqq {\sf F}_i^{-1}(u_i)a_{i,2}({\cal F}_{t-1},\lambda)+a_{i,1}({\cal F}_{t-1},\lambda)$$ and note that for $\lambda_n(h) \coloneqq \lambda_0 + h/\sqrt{n}$, $h \in H \subset {\mathbb R}^d$,
\begin{align*}
\tilde {\mathbb G}^\circ_n(u,s;h) =    \frac1{\sqrt n}\sum_{t=1}^{\floor{sn}}(Z_t\bI\{\eps_{1,t} \leq T_{1,t}(u_1,&\lambda_n(h)),\eps_{2,t} \leq \ T_{2,t}(u_2,\lambda_n(h))\}\\
\,&-\Ex[Z_1{\sf F}\{T_{1,1}(u_1,\lambda_n(h)),T_{2,1}(u_2,\lambda_n(h))\}]).
\end{align*}
Define $$T^\dagger_{i,t}(u_{i},\lambda) \coloneqq {\sf F}_i^{-1}(u_i)a^\dagger_{i,2}({R}_{t},\lambda)+a_{i,1}^\dagger(R_{t},\lambda),$$ with $a^\dagger_{i,1}(R_t,\lambda)\coloneqq \dot{a}_{i,1}(R_t)^\top(\lambda-\lambda_0)$, and $a^\dagger_{i,2}(R_t,\lambda) \coloneqq 1+\dot a_{i,2}(R_t)^\top(\lambda-\lambda_0)$. Let 
\begin{align*}
\tilde {\mathbb G}^{\circ \dagger}_n(u,s;h) :=    \frac1{\sqrt n}\sum_{t=1}^{\floor{sn}}(Z_t\bI\{\eps_{1,t} \leq T_{1,t}^\dagger(u_1,&\lambda_n(h)),\eps_{2,t} \leq \ T^\dagger_{2,t}(u_2,\lambda_n(h))\}\\
\,&-\Ex[Z_1{\sf F}\{T^\dagger_{1,1}(u_1,\lambda_n(h)),T^\dagger_{2,1}(u_2,\lambda_n(h))\}]).
\end{align*}
Recall the definition $E = [0,1]^2 \times [0,1] \times H$. We will show that
\begin{align}\label{eq:Gndiff}
\ssup\limits_{(u,s,h) \in E} |\tilde{\mathbb G}^{\circ \dagger}_n(u,s;h) -\tilde{\mathbb G}_n^\circ(u,s;h)| = o_p(1) 
\end{align}
and that $\tilde{\mathbb G}^{\circ \dagger}_n$ is stochastically equicontinuous.

First, note that 
\begin{align}\label{eq:Gdecomp}
    \tilde{\mathbb G}^{\circ \dagger}_n(u,s;h) -\tilde{\mathbb G}_n^\circ(u,s;h) = X_n(u,s,h)-\bar X_n(u,s,h),
 \end{align}  
 where
 \begin{align*}
X_n(u,s,h) \coloneqq \frac1{\sqrt n}\sum_{t=1}^{\floor{sn}}Z_t(&\bI\{\eps_{1,t} \leq T_{1,t}^\dagger(u_1,\lambda_n(h)),\eps_{2,t} \leq \ T^\dagger_{2,t}(u_2,\lambda_n(h))\}\\
&-\bI\{\eps_{1,t} \leq T_{1,t}(u_1,\lambda_n(h)),\eps_{2,t} \leq \ T_{2,t}(u_2,\lambda_n(h))\})
 \end{align*}  
 and
 \begin{align*}
\bar X_n(u,s,h) \coloneqq
 s_n\sqrt{n}(\Ex[Z_1(&{\sf F}\{T^\dagger_{1,1}(u_1,\lambda_n(h)),T^\dagger_{2,1}(u_2,\lambda_n(h))\}\\ &-{\sf F}\{T_{1,1}(u_1,\lambda_n(h)),T_{2,1}(u_2,\lambda_n(h))\})]).
 \end{align*}
Under Assumption 5, we can show that, 
for some  
$\tilde\lambda_n(h)$ between $\lambda_n(h) = \lambda_0+h/\sqrt{n}$ and $\lambda_0$,  
$$
a_{i,j}({\cal F}_{t-1},\lambda_n(h))=a_{i,j}({\cal F}_{t-1},\lambda_0)+\dot{a}_{i,j}(R_t)^\top \frac{h}{\sqrt n}+\frac1{2n} h^\top \nabla^2_\lambda a_{i,j}({\cal F}_{t-1},\tilde\lambda_n(h))h,
$$
We define $\Delta_n \coloneqq \mmax\limits_{1 \leq i,j \leq 2}\mmax\limits_{1 \le t \leq n}\ssup\limits_{h \in H}|a_{i,j}({\cal F}_{t-1},\lambda_n(h))-a^\dagger_{i,j}(R_t,\lambda_n(h))|,$
and note that
\begin{align*}
\Delta_n \leq \,& \frac1{2n}\ssup\limits_{h \in H}|h|^2 \mmax\limits_{1 \leq i,j, \leq 2} \mmax\limits_{t \leq n} \ssup\limits_{h \in H} |\nabla^2_\lambda a_{i,j}({\cal F}_{t-1},\lambda_n(h))| \\
\leq \,& \frac1{2n}\ssup\limits_{h \in H}|h|^2\mmax\limits_{t \leq n} \ddot{a}_t = O_p(n^{-1+1/q})=o_p(n^{-1/2}),
\end{align*}
where we use that by Assumption~\ref{ass:marginal} and Markov's inequality,  $\mmax\limits_{1 \leq t \leq n} \ddot{a}_t = O_p(n^{1/q})$ for $q > 2$. Hence, for $0\leq \gamma_n = o(n^{-1/2})$, $\mathbb{P}(\Delta_n \leq \gamma_n) \rightarrow 1$ and, on $\{\Delta_n \leq \gamma_n\}$, we may write for any $h \in {H}$
$$
a_{i,j}^\dagger(R_t,\lambda_n(h)) - \gamma_n \leq a_{i,j}({\cal F}_{t-1},\lambda_n(h)) \leq a_{i,j}^\dagger(R_t,\lambda_n(h)) + \gamma_n, \qquad 1 \leq i,j \leq 2,
$$
and, using the monotonicity of the indicator function
$$
\bI\{\eps_{i,t}\leq T^\dagger_{i}(u,\lambda,-\gamma_n;R_t)\} \leq \bI\{\eps_{i,t}\leq T_{i,t}(u,\lambda)\} \leq \bI\{\eps_{i,t}\leq T^\dagger_{i}(u,\lambda,\gamma_n;R_t)\},
$$
$i \in \{1,2\}$, where, akin to \citet[Appendix A.2]{neumeyer2019copula}, we define
\begin{align}
T^\dagger_{i}(u,\lambda,\gamma;R) \coloneqq {\sf F}_i^{-1}(u_i)(a^\dagger_{i,2}({R},\lambda)+\gamma{\sf sign}({\sf F}_i^{-1}(u_i)))+a_{i,1}^\dagger(R,\lambda)+\gamma.
\end{align}
Therefore, on the event $\{\Delta_n \leq \gamma_n\}$, we have 
$$
|X_n(u,s,h)| \leq \frac1{\sqrt n}\sum_{t=1}^{\floor{sn}} (g^{(n)}_{u,h,\gamma_n}(\xi_t)-g^{(n)}_{u,h,-\gamma_n}(\xi_t)),
$$
where $g_{u,h,\gamma}^{(n)}(\xi) = |Z|f_{u,h,\gamma}^{(n)}(\eps,R) \in {\cal G}^{(n)}$,
\begin{align}\label{eq:Gnclass}
{\cal G}^{(n)} \coloneqq   \{\xi = (\eps,R,Z) \mapsto |Z|f^{(n)}(\eps,R): f^{(n)} \in {\cal H}^{(n)},\ i\in \{1,2\}\}
\end{align}
for ${\cal H}^{(n)} \coloneqq {\cal H}_1^{(n)} {\cal H}_2^{(n)}$ and
\begin{align*}
{\cal H}_{i}^{(n)} \coloneqq  \{(\eps,R) \mapsto \bI\{\eps_{i} \leq T_{i}^\dagger(u,\lambda_n(h),\gamma;R)\}:\ & 
(u,h) \in [0,1] \times H, 
\ & |\gamma| \leq \gamma_0, \ \gamma_0 \in (0,\infty)\}.
\end{align*}
Applying a similar argument, we can show that, on $\{\Delta_n \leq \gamma_n\}$,
$$
|\bar X_n(u,s,h)| \leq \sqrt{n} \Ex[g^{(n)}_{u,h,\gamma_n}(\xi_1)-g^{(n)}_{u,h,-\gamma_n}(\xi_1)],
$$
so that we can bound Eq.~\eqref{eq:Gndiff} via
\begin{equation}\label{eq:Gdecomp}
\begin{split}
|\tilde{\mathbb G}^{\circ \dagger}_n(u,s;h) -\tilde{\mathbb G}_n^\circ(u,s;h)|  \leq \,& |\mathbb{Z}_n(s,g^{(n)}_{u,h,\gamma_n})-\mathbb {Z}_n(s,g^{(n)}_{u,h,-\gamma_n})|\\
\,&+2\sqrt{n} \Ex[g^{(n)}_{u,h,\gamma_n}(\xi_1)-g^{(n)}_{u,h,-\gamma_n}(\xi_1)].
 \end{split}   
\end{equation}
where
$$
\mathbb{Z}_n(s,g) \coloneqq \frac1{\sqrt n}\sum_{t=1}^{\floor{sn}}(g(\xi_t)-\Ex[g(\xi_1)]).
$$
We verify Eq.~\eqref{eq:Gndiff}, in two steps by showing that the two terms on the right-hand side of Eq.~\eqref{eq:Gdecomp} vanish uniformly. In particular, we show first that, for the semimetric $\tau((s,f),(t,g)) \coloneqq |s-t|+\rho_r(f,g)$, $\rho_r(f,g) \coloneqq \|f(\xi_1)-g(\xi_1)\|_r$, and some $r\geq 2$, the empirical process $\mathbb{Z}_n(s,g)$, indexed by $(s,g) \in [0,1]\times {\cal G}$, is stochastically equicontinuous
\begin{align}\label{eq:SE}
  \forall \eps>0:\;  \llim\limits_{\delta \downarrow 0}\llimsup\limits_{n\rightarrow \infty}{\mathbb P}\left[\ssup\limits_{f,g \in {\cal G}:\ \tau((s,f),(t,g))\leq \delta} |{\mathbb Z}_n(s,f)-{\mathbb Z}_n(t,g)|>\eps\right] = 0,
\end{align}
where ${\cal G}  \supseteq {\cal G}^{(n)}$ is defined like ${\cal G}^{(n)}$ but with $\lambda$ ranging over $\{\lambda: |\lambda-\lambda_0| \leq \ssup\limits_{h\in H}|h|\}$. Second,  we show that 
\begin{align}\label{eq:SE1}
\delta_n \rightarrow 0,  \quad \delta_n \coloneqq \ssup_{|\gamma|\leq \gamma_n}\ssup_{u,h} 
\rho_r(g^{(n)}_{u,h,\gamma},g^{(n)}_{u,h,0}).
\end{align}
Taken together, this implies
\begin{align}\label{eq:SEGgamma}
\ssup\limits_{|\gamma| \leq \gamma_n} \ssup\limits_{(u,s,h) \in T} | \mathbb {Z}^{(n)}_n(s,g_{u,h,\gamma})-\mathbb {Z}_n(s,g^{(n)}_{u,h,0})| = o_p(1), \quad g^{(n)}_{u,h,\gamma} \in {\cal G}^{(n)},
\end{align}
and $\sqrt{n} \Ex[g^{(n)}_{u,h,\gamma_n}(\xi_1)-g^{(n)}_{u,h,0}(\xi_1)] = o(1)$. Thus, in view of Eq.~\eqref{eq:Gdecomp}, Eq.~\eqref{eq:Gndiff} is proven. Finally, the stochastic equicontinuity of $\mathbb{Z}_n(s,g)$, $(s,g) \in [0,1]\times {\cal G}^{(n)}$ implies that of $\tilde {\mathbb G}_n^{\circ \dagger}(u,s;h) = \mathbb {Z}_n(s,\tilde g^{(n)}_{u,h,0})$, $\tilde g^{(n)}_{u,h,\gamma} = Zf^{(n)}_{u,h,\gamma}$, $f^{(n)}_{u,h,\gamma} \in {\cal H}^{(n)}$. It thus remains to prove (1) Eq.~\eqref{eq:SE} and (2) Eq.~\eqref{eq:SE1}. 

For (1), we apply \citet[Cor 2.7]{mohr2020weak} (see also \citealp[Cor 1]{scholze2024weak}) which is an extension of \cite{andrews1994introduction} to the sequential case and allowing for unbounded classes of functions. We start with the following definition:

\begin{definition}[Bracketing Number] For all $\eps>0$, let $N = N(\eps)$ be the smallest integer, for which there exist a class of functions ${\cal X} \longrightarrow {\mathbb R}$, denoted by ${\cal B}$ and called bounding class, and a function class ${\cal A} \subset {\cal F}$, called approximating class, such that $|{\cal A}|=|{\cal B}|=N$, $\rho(b)<\eps$, $\forall b \in {\cal B}$ and for all $f \in {\cal F}$, there exists $(a,b) \in {\cal A} \times {\cal B}$ such that $|f-a|\leq b$. Then $N=N_{[\ ]}(\eps,{\cal F},\rho)$ is called the bracketing number.    
\end{definition}

By \citet[Cor 2.7]{mohr2020weak} (see also \citealp[Cor 1]{scholze2024weak}), Eq.~\eqref{eq:SE} holds with $\rho_{\nu(2+\lambda)/2}(g) = \|g(\xi_1)\|_{\nu(2+\lambda)/2}$, if there are $\lambda>0$ and $\nu \geq 2$, $\nu \in {\mathbb N}$, such that: 

\begin{itemize}
    \item[\bf SE1] There exists $\lambda>0$ and an integer $\nu\geq 2$ such that $\sum_{s=1}^\infty s^{\nu-2}\alpha(s)^{\frac{\lambda}{2+\lambda}} < \infty$.
    \item[\bf SE2]  For $\rho_2 \coloneqq \|f(\xi_1)\|_2$ and $(\nu,\lambda)$ from SE1
    $$\int_0^1 \eps^{-\frac{\lambda}{2+\lambda}} N^{\frac{1}{\nu}}_{[\ ]}(\eps,{\cal G},\rho_2) \d \eps < \infty,$$   for each $\eps>0$
    $$
    \forall b\in{\cal B}:\;\Ex[|b(\xi_1)|^{l\frac{2+\lambda}{2}}]^{\frac1{2}}\leq \eps, \quad l \in \{2,\dots,\nu\},
    $$
    and for some $c \in [1,\infty)$
    $$
    \ssup\limits_{g\in {\cal G}}\Ex[|g(\xi_1)|^{\nu\frac{2+\lambda}{2}}] \leq c.
    $$  
\end{itemize}

Without of loss of generality, assume $Z$ to be a scalar and, for completeness, define class ${\cal G} = |Z|\cal{H}$, ${\cal H} = {\cal H}_1{\cal H}_2$, with
\begin{align*}
{\cal H}_{i} \coloneqq  \{(\eps,R) \mapsto \bI\{\eps_{i} \leq T_{i}^\dagger(u,\lambda,\gamma;R)\}:\ & u \in [0,1],\ \lambda \in \{l: |l-\lambda_0| \leq \ssup\limits_{h\in H}|h|\}, \\
\ & |\gamma| \leq \gamma_0, \ \gamma_0 \in (0,\infty)\}.
\end{align*}
Clearly, ${\cal G}^{n} \subset {\cal G}$. Now, note that $N_{[\ ]}(\eps,{\cal H}, \rho_2) = O(\eps^{-2(d+3)})$. This can be shown using similar arguments as those used in the proof of \citet[Lemma S.1 ($a$)]{mw:2023}.  Hence, for any $\eta > 0$, there exists functions $(a_r,b_r)_{r=1}^{N}$ such that for any $h\in {\cal H}$ there exists an $r$ such that $|h-a_r|\leq b_r$, where $\rho(b_r)\leq \eta$ for any $r \in \{1,\dots, N\}$, $N \coloneqq N_{[\ ]}(\eps,{\cal H}, \rho_2)$. Because ${\cal H}$ is an indicator class of functions, we have $b_r \in [0,1]$. Now, define $A_r = Z a_r$, $B_r = |Z| b_r$. Then for any $h \in {\cal H}$ satisfying $|h-a_r|\leq b_r$ we have $|Zh-A_r|\leq B_r$.  Now, for $p>2$, $q = p/(p-2)$ we get for $\eta = (\eps/\|Z_1\|_{p})^q$:
$$
\rho(B_r) \leq \|Z_{1}\|_{p} \|b_r(\xi_1)\|_{2q} \leq \|Z_{1}\|_{p} \|b_r(\xi_1)\|_{2}^{\frac1{q}} \leq \|Z_1\|_{p} \eta^{\frac1{q}} = \eps,
$$
where the penultimate inequality uses $b_r \in [0,1]$. This hows that for $N_1 = N_{[\ ]}(\eps,{\cal G},\rho) = O(\eps^{-2(d+3)q})$, $(A_r,B_r)_{r=1}^{N_1}$ are valid brackets for ${\cal G}$, with $\rho(B_r)\leq \eps$ for any $r \in \{1,\dots, N_1\}$. Clearly, ${\bf SE1}$ holds for any $(\nu,\lambda)$. Moreover, we know that if $\|Z_1\|_p$, $p>2$, then $N_{[\ ]}(\eps,{\cal G},\rho) = O(\eps^{-v})$, for $v \coloneqq aq$, $a\coloneqq 2(d+3)$ and $q \coloneqq p/(p-2)$. The entropy integral in {\bf SE2} is finite if $\lambda/(\lambda+2)+v/\nu < 1$, which, as $\lambda \downarrow 0$, reduces to $\nu > v$. As $|g(\xi)| \leq |Z|$ for all $g \in {\cal G}$, the moment condition in {\bf SE2} is satisfied if $\Ex[|Z_1|^{\nu(2+\lambda)/2}]<\infty$, or $\nu(2+\lambda)\leq 2p$. Now, assume $p>a+2$ and choose an integer $\nu$ such that $v < \nu \leq p$. Since $\nu > v$, we can pick $\lambda>0$ small such that $\lambda/(\lambda+2)+v/\nu < 1$ holds and, at the same time, $\nu(2+\lambda)\leq 2p$. This proves Eq.~\eqref{eq:SE}, i.e. stochastic equicontinuity of ${\mathbb Z}(s,g)$ indexed by $(s,g) \in [0,1]\times {\cal G}$.

Next, we show Eq.~\eqref{eq:SE1} for $r = \nu(\lambda+2)/2$, as chosen previously. Here, we first verify that
there exist constants $\gamma_0 \in (0,1/2)$, $c\in (0,\infty)$ and an event $\Omega_n$ with $\mathbb{P}(\Omega_n) \rightarrow 1$ such that on $\Omega_n$ with $|\gamma| \leq \gamma_0$:
\[
\mmax\limits_{1 \leq i \leq 2}\mmax\limits_{t\leq n}\ssup\limits_{(u,h) \in [0,1] \times H}|{\sf F}_i(T_{i}^\dagger(u,\lambda_n(h),\gamma;R_t))-{\sf F}_i(T_{i}^\dagger(u,\lambda_n(h),0;R_t))| \leq c|\gamma|,
\]
for some $c \in (0,\infty)$. To see this, consider $(y,v,w,s) \in {\mathbb T} \coloneqq  {\mathbb R} \times [1/2,3/2] \times [-1,1] \times \{-1,1\}$. There exists $x = x(y,v,w,\gamma)$ between $yv+w$ and $y(v+s\gamma)+w+\gamma$ such that
\begin{align*}  
{\sf F}_i (y(v+s&\gamma)+w+\gamma) -{\sf F}_i(yv+w) = \gamma f_i(x)(1+ys), \quad (y,v,w,s) \in {\mathbb T}.
\end{align*}
Moreover, for some $\gamma_0 \in (0,1/2)$, there exist $c_0,c \in (0,\infty)$ such that
\[
\ssup\limits_{|\gamma| < \gamma_0}\ssup\limits_{(y,v,w) \in  {\mathbb T}}f_i(x)(1+|y|) \leq  c_0 \ssup\limits_{z \in \mathbb R}f_i(z)(1+|z|) \leq c.
\]
To see this, note that for some $\theta \in (0,1)$, $x = y(v+\theta s\gamma)+(w+\theta \gamma)$. Now, by the reverse triangle inequality, $|v+\theta s\gamma| \geq v - |\gamma| \geq 1/2-\gamma_0$ and $|w+\theta\gamma| \leq 1+\gamma_0$. Therefore, $|x| \geq |\frac1{2}-\gamma_0||y|-(1+\gamma_0)$, or $1+|y| \leq c_0(1+|x|)$, for $c_0 \coloneqq 2(2+\gamma_0)/(1-2\gamma_0)$ and all $y \in {\mathbb R}$. The final inequality is due to Assumption~\ref{ass:epsdist}. Next, since $\Ex[|\dot a_{i,j}(R_1)|^q]<\infty$, with $q > 2$, Assumption~\ref{ass:marginal} yields for $i \in \{1,2\}$:
$$\mmax\limits_{t\leq n}\ssup\limits_{h \in H}|a_{i,1}^\dagger(R_t,\lambda_n(h))| = o_p(1), \quad \mmax\limits_{t\leq n}\ssup\limits_{h \in H}|a_{i,2}^\dagger(R_t,\lambda_n(h))-1| = o_p(1).$$
In particular, $\mathbb{P}(\Omega_n) \rightarrow 1$, where 
\begin{align}\label{eq:Omegan}
 \Omega_n \coloneqq \Omega_{1,n} \cap \Omega_{2,n},\quad\Omega_{1,n} \coloneqq \{\mmax\limits_{t\leq n, 1\leq i \leq 2}\ssup\limits_{h \in H}|a_{i,1}^\dagger(R_t,\lambda_n(h))| \leq 1\},
\end{align}
and
$$
 \Omega_{2,n} \coloneqq \{\mmax\limits_{t\leq n,1\leq i \leq 2}\ssup\limits_{h \in H}|a_{i,2}^\dagger(R_t,\lambda_n(h))-1| \leq \frac1{2}\}.
$$
Therefore, with $f^{(n)}_{u,h,\gamma}  \in \mathcal{H}^{(n)}$, we get for $|\gamma| \leq \gamma_0$ and some $c \in (0, \infty)$
\[
\ssup\limits_{(u,h) \in [0,1] \times H}\mathbb{P}(\{|f^{(n)}_{u,h,\gamma}-f^{(n)}_{u,h,0}|>0\}, \Omega_n) \leq c|\gamma|.
\]
Now, for $\gamma\geq 0$ and any $0<m \leq \nu(2+\lambda)/2$, we have
$$
\Ex[|Z_t|^m |f^{(n)}_{u,h,\gamma}-f^{(n)}_{u,h,0}|^m]=\Ex[|Z_t|^m \Ex[f^{(n)}_{u,h,\gamma}-f^{(n)}_{u,h,0}\mid {\cal F}_{t-1}]],
$$
using the monotonicity of the indicator and that  $\Ex[|Z|^m]<\infty$ because $\nu(2+\lambda)/2 \leq p$, with $p$ from Assumption~\ref{ass:Z}. Moreover, from the previous discussion,
$$
\Ex[|Z_t|^m \Ex[f^{(n)}_{u,h,\gamma_n}-f^{(n)}_{u,h,0}\mid {\cal F}_{t-1}]] \leq c|\gamma_n|\Ex[|Z_t|^m]+\Ex[|Z_t|^m\bI\{\Omega_n^c\}].
$$
Thus, $\mathbb{P}(\Omega_n^c)\rightarrow 0$ and the dominated convergence theorem imply $\Ex[|Z_t|^m\bI\{\Omega_n^c\}] \rightarrow 0$. Thus, $\rho_m(g_{u,h,\gamma}^{(n)},g_{u,h,0}^{(n)}) \leq c|\gamma|^{1/m}+o(1)$. Letting $\gamma=\gamma_n \rightarrow 0$ finishes the proof. \hfill $\square$

\subsection{Proof of Lemma 1 (2)} 
Under the null, $\mathbb{P}\{U_t \leq u \mid {\cal F}_{t-1}\} = {\sf C}(u)$; i.e., $\bI\{U_t \leq u\}-{\sf C}(u)$ is a MDS wrt ${\cal F}_t$. Since $\tilde{\mathbb C}_n(u,s;\lambda_0) = \frac1{\sqrt n}\sum_{t=1}^{\floor{sn}}(\bI\{U_t \leq u\}-{\sf C}(u))$ and $$\Ex[\tilde G_n(u,s;\lambda_0)] = \Ex[Z_t{\sf F}(T_{1,t}(u_1,\lambda_0),T_{2,t}(u_2,\lambda_0))] = \mu {\sf C}(u),$$ we can write
\begin{align*} 
\tilde{\mathbb G}^\circ_n(u,s;\lambda_0)=\tilde{\mathbb G}_n(u,s;\lambda_0) =  \mu\tilde{\mathbb C}_n(u,s;\lambda_0) + \tilde {\mathbb B}_{n,Z}(u,s),
\end{align*}
where $\tilde {\mathbb B}_{n,Z}(u,s) \coloneqq \frac1{\sqrt n}\sum_{t=1}^{\floor{sn}}(Z_t-\mu)\bI\{U_t \leq u\}$. By Lemma~\ref{lem:seqtilde} ($1$), this is a sum of two tight sequential empirical processes.

Moreover, for a given $(s,u) \in [0,1] \times [0,1]^2$, $\tilde{\mathbb C}_n(u,s;\lambda_0)$ is a sum of MDS, while $\tilde {\mathbb B}_{n,Z}(u,s)$ is a sum of strictly stationary and strongly mixing processes. By the continuous mapping theorem, the claim follows from the fidi-convergence of the process $(\tilde{\mathbb C}_n(s_r,u_r),\tilde{\mathbb B}_{Z,n}(s_r,u_r))_{r=1}^m$, for a fix $m \in {\mathbb N}$, to a tight Gaussian random vector. This is done using the Cramér–Wold device and the CLT for arrays of strongly mixing processes, applied to
$
\sum_{r=1}^m a_r\tilde{\mathbb C}_n(s_r,u_r) + b_r^\top \tilde{\mathbb B}_{Z,n}(s_r,u_r) = \frac1{\sqrt n}\sum_{t=1}^n Y_{t,n},
$
where
$$
Y_{t,n} \coloneqq \sum_{r=1}^m \bI\{t \leq \floor{s_rn}\}(a_r(\bI\{U_t \leq u_r\}-{\sf C}(u_r)) + b_r^\top(Z_t-\mu)\bI\{U_t \leq u_r\}),
$$
with $(a_r,b_r,s_r,u_r) \in {\mathbb R} \times {\mathbb R}^d \times [0,1] \times [0,1]^2$. Since, for some $c_1,c_2 \in (0,\infty)$, $|Y_{t,n}| \leq c_1+c_2|Z_t-\mu|$, it follows from Assumption~\ref{ass:Z} that $\ssup_{t,n}\Ex[|Y_{t,n}|^{2+\delta}]<\infty$ for some $\delta>0$. Also, by independence of $U_t$ form ${\cal F}_{t-1}$, $Y_{t,n}$ inherits the (geometric) mixing rate from $(Z_t,U_t)$. Hence, \citet[Thm 4.4]{rio2017asymptotic} applies and 
$\frac1{\sqrt n}\sum_{t=1}^n Y_{t,n} \rightarrow_d \mathcal{N}(0,\sigma^2)$, where $\sigma^2 = \sum_{q,r=1}^{m} (s_r \wedge s_q) \sigma_{rq},$
with
\[
\sigma_{rq} = a_ra_q \sigma_C(u_r,u_q)+a_r \sigma_{CZ}(u_r,u_q)b_q + b_r^\top \sigma_{CZ}(u_r,u_q)^\top a_q  + b_r^\top \sigma_{Z}(u_r,u_q) b_q   \geq 0.
\]
This proves the claim. \hfill $\square$
\subsection{Proof of Lemma 1 (3)} 
By the proof of Part 1, it suffices to show that 
$$
\ssup\limits_{(u,s,\lambda) \in T   } |{\mathbb G}^{\dagger}_n(u,s,\lambda_0+h/\sqrt n)-{\mathbb G}^\dagger_n(u,s,\lambda_0)-s\Psi^\top(u)h| = o_p(1).
$$
Define $$
g(u,\lambda; Z,R) \coloneqq Z {\sf F}(T^\dagger_{1}(u_1,\lambda;R),T^\dagger_{2}(u_2,\lambda;R)), \quad  m(u,\lambda) \coloneqq \Ex[g(u,\lambda; Z_1,R_1)]. 
$$   
Note that, by the law of iterated expectations and $\eps_t \perp {\cal F}_{t-1}$,
$${\mathbb G}_n^\dagger(u,s;\lambda)  - {\mathbb G}_n^{\circ \dagger}(u,s;\lambda)= \sqrt{n}s_n(m(u,\lambda)-\mu {\sf C}(u)),$$ while, by Part 1, we have $
 \ssup_{u,s,h} | {\mathbb G}_n^{\circ \dagger}(u,s;\lambda_0+h/\sqrt n)  - {\mathbb G}_n^{\circ \dagger}(u,s;\lambda_0) | = o_p(1).$
Moreover, $m(u,\lambda_0) =\mu {\sf C}(u)$, so that,   
 $$
 \mathbb{G}^\dagger_n(u,s;\lambda_0+h/\sqrt n)-\mathbb{G}^\dagger_n(u,s;\lambda_0)  =  s\sqrt{n}(m(u,\lambda_0+h/\sqrt n) -m(u,\lambda_0)) + o_p(1),
 $$
uniformly in $(u,s,h) \in T$. Hence, it suffices to show the deterministic expansion
$$
\ssup\limits_{(u,h) \in [0,1]^2\times H} |\sqrt{n}(m(u,\lambda_0+h/\sqrt n)-m(u,\lambda_0))-\Psi^\top h|=o(1).
$$
We do so in two steps 1) $\nabla m(u,\lambda_0) = \Psi(u)$ and 2) apply a second-order Taylor expansion. 

Step 1) and 2) make use of the following result that is proven below: On the sequence of events $\Omega_n$ defined in Eq.~\eqref{eq:Omegan}, it holds:
\begin{align}\label{eq:nmbound}
    \ssup\limits_{(u,h) \in[0,1]^2 \in H} \mmax\{\|\nabla_\lambda g(u;\lambda_n(h),Z,R)\|, \|\nabla^2_\lambda g(u;\lambda_n(h),Z,R)\|\} \leq c\dot g(Z,R),
\end{align}
with $c \in (0,\infty)$ and  $\dot g(Z,R) \coloneqq |Z|(1+\mmax\limits_{1\leq i,j \leq 2}|\dot a_{i,j}(R)|^2)$, $\Ex[\dot g(Z_1,R_1)] < \infty$.

{\it Step 1}. 
By the chain-rule
\begin{align*} 
\nabla_\lambda g(u;\lambda,Z,R) =  Z\sum_{k \in \{1,2\}}\partial_k {\sf F}\{T^\dagger_{1}(u_1,\lambda;R),T^\dagger_{2}(u_2,\lambda;R)\} \xi_k(u_k,R). 
\end{align*}
At $\lambda=\lambda_0$, we have $T^\dagger_{k}(u_k,\lambda;R) = {\sf F}_k^{-1}(u_k)$, $\partial_k {\sf F}\{T^\dagger_{1}(u_1,\lambda;R),T^\dagger_{2}(u_2,\lambda;R)\} = \partial_k {\sf C}(u)f_k({\sf F}_k^{-1}(u_k))$. From
$
\xi_k(u,R) =  {\sf F}_k^{-1}(u_k) \dot a_{k,2}(R)+ \dot a_{k,1}(R)
$
we get
$$
\nabla_\lambda g(u;\lambda_0,Z,R)   = \sum_{k \in \{1,2\}}\partial_k {\sf C}(u)f_k({\sf F}_k^{-1}(u_k)) Z\xi_k(u_k,R).
$$
If $\nabla_\lambda m(u,\lambda) = \Ex[\nabla_\lambda g(u,\lambda;Z_1,R_1)]$, i.e if  we can justify differentiation under the expectation, then $\nabla_\lambda m(u,\lambda_0) = \Psi(u)$. This follows, however, from Eq.~\eqref{eq:nmbound} and the dominated convergence theorem.

{\it Step 2}. By a second-order Taylor expansion and Eq.~\eqref{eq:nmbound}, it holds uniformly
\begin{align*}  
  \sqrt{n}(m(u,\lambda_n(h)) -m(u,\lambda_0)) =   \Psi(u)h + \frac{1}{2\sqrt n}  h^\top (\nabla^2 m(u,\tilde\lambda)) h  = \Psi(u)h + o_p(1),
\end{align*}  
for some $\tilde\lambda$ between $\lambda_0$ and $\lambda_0+h/\sqrt n$.  

Proof of Eq.~\eqref{eq:nmbound}. First, note that on $\Omega_n$ we have $|a_{k,1}^\dagger(R,\lambda)|\leq 1$ and $|a_{k,2}^\dagger(R,\lambda)| \in [1/2,3/2]$. So, with $y_k \coloneqq {\sf F}^{-1}_k(u_k)$ the reverse triangle inequality yields:
$$
|T_k^\dagger(u_k,\lambda;R)| \geq |y_k||a_{k,2}^\dagger(R,\lambda)|-|a_{k,1}^\dagger(R,\lambda)| \geq \frac1{2}|y_k|-1.
$$
Now, since $|\partial_k {\sf F}(x_1,x_2)| \leq f_k(x_k)$, we get
$$
|\nabla_\lambda g(u,\lambda;Z,R)| \leq |Z| \sum_{k=1}^2 |f_k(T_k^\dagger(u_k,\lambda;R))||\xi_k(u_k,R)|
$$
By Assumption~\ref{ass:epsdist}, $\ssup_y (1+|y|)f_k(y) < \infty$. So, for $1\leq k \leq 2$
$$f_k(T_k^\dagger(u_k,\lambda;R))|y_k|\leq 2f_k(T_k^\dagger(u_k,\lambda;R))(1+|T_k^\dagger(u_k,\lambda;R)|) \leq c < \infty.
$$
Hence, on $\Omega_n$, there is a constant $c < \infty$ such that
$$
|\nabla_\lambda g(u,\lambda;Z,R)| \leq c |Z| \mmax\limits_{1\leq i,j \leq 2}|\dot a_{i,j}(R)|.
$$
Similarly,
$$
|\nabla^2_\lambda g(u,\lambda;Z,R)| \leq  |Z| \sum_{i,j=1}^2| \partial_i\partial_j{\sf F}(T_1^\dagger(u_1,\lambda;R),T_1^\dagger(u_1,\lambda;R))||\xi_i(u_i,R)||\xi_j(u_j,R)|
$$
Using that, by Assumption~\ref{ass:marginal}, 
$$
| \partial_i\partial_j{\sf F}(x_1,x_2)| \leq \frac{c}{(1+|x_i|)(1+|x_j|)}
$$
it can be equivalently shown that for some constant $c \in (0, \infty)$:
$$
|\nabla^2_\lambda g(u,\lambda;Z,R)| \leq c |Z| \mmax\limits_{1\leq i,j \leq 2}|\dot a_{i,j}(R)|^2.
$$
The claim thus follows from Hölders inequality.

\section{Proof of Lemma 2}

Recall from Eq.~\eqref{eq:gn_tilde_rel} that we can recover the pair $(G_n, C_n)$ from the pair $(\tilde G_n,\tilde C_n)$ via the functional map:
$$
(G_n,C_n)(u,s;\lambda) = (\tilde G_n,\tilde C_n)(\tilde C_{1,n}^{-}(u_1,s;\lambda), \tilde C_{2,n}^{-}(u_2,s;\lambda),s; \lambda).
$$
The properties of $\tilde {\mathbb H}_n(u,s;\lambda_0+h/\sqrt{n})=\sqrt{n}s_n((\tilde G_n,\tilde C_n)-(\mu {\sf C},{\sf C}))(u,s;\lambda_0 + h/\sqrt{n})$ have been established under the local drift $h \in H \subset \mathbb{R}^d$. If we can show that the mapping underlying Eq.~\eqref{eq:gn_tilde_rel} obeys the extended continuous mapping theorem, then the weak limit of ${\mathbb H}_n(u,s;\hat\lambda)$ claim follows because $\tilde {\mathbb H}^\circ_n(u,s;\lambda_0)=\tilde {\mathbb H}_n(u,s;\lambda_0)$, and, by assumption, $(\tilde {\mathbb H}^\circ_n(u,s;\lambda_0),h_n) \Rightarrow (\mathbb H, \Lambda)$, where the random drift $h_n \coloneqq \sqrt{n}(\lambda_n-\lambda_0)$ satisfies, by Assumption~\ref{ass:marginal}, ${\mathbb P}(h_n \in H) \rightarrow 1$. It thus remains to establish  the extended continuous mapping theorem. In what follows, we focus on $\mathbb G_n$, as $\mathbb C_n$ is a special case.

Define $\Phi_n: {\cal A}_n \longrightarrow \ell^\infty(E;{\mathbb R}^k)$, $E = [0,1]^2 \times [0,1] \times H$, via $s_nG_n = \Phi_n(s_n \tilde G_n,s_n \tilde C_n).$ We will first show that $\Phi_n$ is sufficiently smooth at $s_n(\mu {\sf C},{\sf C})$. The proof follows closely the proof of Theorem 2.4 in \cite{bucher2013empirical} and the proof of Lemma B.1 in \citet[Appendix B]{bucher2016dependent}. The scaling by $s_n$ ensures uniformity over $s \in [0,1]$. To begin with, recall $(u,s,h) \in E$, and, $u = (u_1,u_2)$, $u^{(1)} = (u_1,1)$, $u^{(2)} = (1,u_2)$. Define the domain ${\cal A}_n \coloneqq \ell^\infty(E;\mathbb{R}^k) \times {\cal E}_n^\star$, 
\[
{\cal E}_n^\star \coloneqq \{K\in \ell^\infty(T): \forall h \in H,\ 
(u,s)\mapsto K(u,s,h)\in {\cal E}_n\},
\]
where ${\cal E}_n$ is the set of all bounded function $F: [0,1]^2 \times [0,1]\mapsto {\mathbb R}$, such that for all $s\in [0,1]$, $u \mapsto F(u,s)$ is a bivariate cdf on $[0,1]^2$ with total mass $s_n$, $F((1,1),s)=s_n$, margins $F_p(u_p,s) = F(u^{(p)},s)$ so that $F_p(0)=0$ and $F_p(1,s) = s_n$. We note that $(s_n \tilde G_n,s_n \tilde C_n) \in {\cal A}_n$ $a.s.$. Next, for any non-decreasing function $F(\cdot, s,h)$ with values in $[0,s_n]$ define the {\it scaled} generalized inverse
$$
{\sf in}_n(F)(u,s,h) \coloneqq \iinf\{x \in [0,1]: F(x,s,h) \geq s_nu\},\quad {\sf in}_n(F)(u,0,h) = 0.
$$
Finally, let $\Phi_n: {\cal A}_n \longrightarrow \ell^\infty(E;\mathbb{R}^k)$, where
\[
\Phi_n(G,K)(u,s,h) = G({\sf in}_n(K_1)(u_1,s,h),{\sf in}_n(K_2)(u_2,s,h), s,h),
\]
with the convention $\Phi_n(G,K)(u,0,h) \coloneqq 0$. In view of Eq.~\eqref{eq:gn_tilde_rel}, 
$$(G,K)(u,s,h) = s_n(\tilde G_n,\tilde C_n)(u,s;\lambda_0+h/\sqrt{n}),$$ so that for $p \in \{1,2\}$: ${\sf in}_n(K_p)(u_p,s,h) = \tilde C_{p,n}^{-}(u_p,s;\lambda_0+h/\sqrt{n})$ and hence $s_n G_n = \Phi_n(s_n \tilde G_n,s_n \tilde C_n)$. The map is evaluated at the {\it limit point} $(G_{0,n},K_{0,n})$, with $G_{0,n}(u,s,h)=s_n\mu{\sf C}(u)$ and $K_{0,n}(u,s,h)=s_n{\sf C}(u)$, independent of $h$, and $K_{0,n} \in {\cal E}_n^\star$ with $G_{0,n} = \Phi_n(G_{0,n},K_{0,n})$.

Now, define the tangential set ${\cal T}_0^\star \coloneqq {\cal E}_0^\star \times {\cal K}_0^\star \subset C(T;\mathbb{R}^k) \times C(T)$, with $(G_0,K_0) \in {\cal T}^\star$. ${\cal E}_0^\star$ consists of continuous function $g: E \rightarrow {\mathbb R}^k$ satisfying $$g(u,0,h)=g((u_1,0),s,h)=g((0,u_2),s,h)=0.$$ ${\cal K}_0^\star$ consists of continuous functions $\kappa: T \rightarrow {\mathbb R}$ such that $$\kappa(u,0,h)=\kappa((u_1,0),s,h)=\kappa((0,u_2),s,h)=\kappa((1,1),s,h)=0$$ and $\kappa(u^{p},s,h) = 0$ whenever $u_p \in \{0,1\}$.  Set for some $(g,\kappa) \in {\cal T}^\star$ 
$$
\Upsilon_n(g,\kappa) \coloneqq \sqrt{n}(\Phi_n(G_{0,n}+n^{-1/2}g,K_{0,n}+n^{-1/2}\kappa)-\Phi_n(G_{0,n},K_{0,n})).
$$
We are now ready to state the main result.\footnote{As a direct consequence of Lemma~\ref{lem:sec_had}, we obtain $\ssup\limits_{(u,s,h)\in T}|\Upsilon_n(\kappa_n,\kappa_n)-\Upsilon_C(\kappa)|\rightarrow 0$ for the special case $Z = 1$ ($\mu = 1$)
$$
\Upsilon_C(\kappa)(u,s,h) \coloneqq \kappa(u,s,h)- \sum_{p=1}^2 \partial_p{\sf C}(u)\kappa(u^{(p)},s,h). 
$$} 

\begin{lemma}\label{lem:sec_had} If $(g_n,\kappa_n)\rightarrow (g,\kappa) \in {\cal T}^\star$ uniformly in $E$, then $$\ssup\limits_{(u,s,h)\in E}|\Upsilon_n(g_n,\kappa_n)-\Upsilon(g,\kappa)|\rightarrow 0,$$ where
$$
\Upsilon(g,\kappa)(u,s,h) \coloneqq g(u,s,h)-\mu \sum_{p=1}^2 \partial_p{\sf C}(u)\kappa(u^{(p)},s,h). 
$$
\end{lemma}

Lemma~\ref{lem:sec_delt} is a consequence of the extended continuous mapping theorem (see, e.g., \citealp{van1994weak}), Lemma~\ref{lem:sec_had}, and Lemma~\ref{lem:seqtilde}. In particular, set $g_n(u,s,h) = \tilde {\mathbb G}_n(u,s,\lambda_0+h/\sqrt{n})$, $\kappa_n(u,s,h) = \tilde {\mathbb C}_n(u,s,\lambda_0+h/\sqrt{n})$. Then $g(u,s,h) = \mu {\mathbb B}_C(u,s)+ {\mathbb B}_Z(u,s)+s\Psi(u)h \in {\cal E}^\star$, $\kappa(u,s,h) =  {\mathbb B}_C(u,s)+s\psi_C(u)h \in {\cal K}^\star$. Therefore, by Lemma~\ref{lem:sec_delt}, we get 
\begin{align*}
   \mathbb{G}_n(\cdot,\cdot;\lambda_0+\cdot/\sqrt{n}) =  \Upsilon_n(g_n,\kappa_n)\Rightarrow \Upsilon(g,\kappa),
\end{align*}
in $\ell^\infty(E;{\mathbb R}^k)$. The claim then follows from the extended continuous mapping theorem because $(\mathbb G_n(\cdot,\cdot;\lambda_0), \sqrt{n}(\hat\lambda-\lambda_0))\Rightarrow ({\mathbb G},\Lambda)$.

It thus remains to prove Lemma~\ref{lem:sec_had}.

{\it Proof of Lemma~\ref{lem:sec_had}.} The proof follows almost verbatim \citet[Lem B.1]{bucher2016dependent}. There are two extensions: 1) We enlarge the index set from $(u,s)$ to $(u,s,h)\in[0,1]^2\times[0,1]\times H$, where $H\subset\mathbb{R}^{d}$ is compact, in order to accommodate local deviations $\lambda=\lambda_0+h/\sqrt{n}$ and thus control the first-step estimation error. This does not create new difficulties: the argument in \citet[Lem.~B.1]{bucher2016dependent} is purely deterministic and based on uniform bounds. Hence it can be applied pointwise in $h$ and remains valid after taking an additional supremum over $h\in H$. 2) We work with the two-input composition map $(G,K)\mapsto \Phi_n(G,K)$ and apply it to
$(s_n\tilde G_n,\,s_n\tilde C_n)$. The marked case is handled componentwise since $G$ takes
values in $\mathbb{R}^k$. The only delicate step concerns the control of the inverse terms
${\sf in}_n(K_p)$ (equivalently, the generalized inverses of the marginal sections of
$\tilde C_n$), which is exactly the issue treated in \citet[Lem.~B.1]{bucher2016dependent}.
No further complications arise beyond tracking the additional indices.

\section{Proof of Proposition 1}
\subsection{Proof of Proposition 1 (1)}
To begin with, define
 $$
 \tilde\ell(u,\alpha) \coloneqq \frac1{n}\sum_{t=1}^n \ell_t(u,\alpha).
 $$
 We will show that, uniformly in $(u,\alpha) \in {\cal U}\times {\cal A}$, ($1$) $|\hat\ell(u,\alpha)-\tilde \ell(u,\alpha)|=o_p(1)$ and ($2$) $|\tilde \ell(u,\alpha)-\ell(u,\alpha)|=o_p(1)$. Now, if (1) and (2) are true, the claim follows from the uniform separation result in Eq.~\eqref{eq:separation}. Note that
 \begin{align*}
 \ssup\limits_{(u,\alpha) \in {\cal U} \times {\cal A}} |\hat\ell(u,\alpha)-\tilde \ell(u,\alpha)| \leq \,& \mmax_{1\leq j \leq 4} \ssup\limits_{(u,\alpha) \in {\cal U} \times {\cal A}} | \llog p_j(u,\alpha)|\\
 \,& \quad \ssup\limits_{u \in {\cal U}}  \sum_{j=1}^4\big|\frac1{n}\sum_{t=1}^n \hat d_{j,t}(u) -\frac1{n}\sum_{t=1}^n d_{j,t}(u)\big|,
  \end{align*}
  where we recall that $\hat d_{j,t}(u)\coloneqq \hat d_{j,t}(u,1)$ and $d_{j,t}(u)$ are the quadrant indicators based on the true ranks. Because $\iinf_{u,a} p_j(u,\alpha)  > 0$, we get $\ssup_{u,\alpha} |\llog p_j(u,\alpha)| < \infty$. Moreover, there exists a $c \in (0,\infty)$ such that, by Lemma~\ref{lem:seqtilde} and \ref{lem:sec_delt}, 
\begin{align} \label{eq:dhatdiff}
    \ssup\limits_{u \in {\cal U}}  |\frac1{n}\sum_{t=1}^n \hat d_{j,t}(u) -\,&\frac1{n}\sum_{t=1}^n d_{j,t}(u)\big| \nonumber \\ \leq \,&\frac{c}{\sqrt n} \ssup\limits_{u \in {\cal U}} \sqrt{n}|C_n(u,1;\hat\lambda)-\tilde C_n(u,1;\lambda_0)| = O_p(n^{-1/2}).
\end{align}
Finally, we apply  \citet[Lemma 2.4]{newmc:94} to show (2). This follows because: ($i$) $\ell_t(u,\alpha)$ is a sequence of {\sf IID} random variables, ($ii$) $(u,\alpha) \mapsto \ell_t(u,\alpha)$ is continuous, ($iii$) $\ssup_{u \in {\cal U}}\ssup_{\alpha \in {\cal A}}|\ell_t(u,\alpha)| \leq c$,  ($iv$) and ${\cal U} \times {\cal A}$ is compact. This completes the proof. \hfill $\square$

\subsection{Proof of Proposition 1 (2)}

By definition, $\hat\alpha(u)$ is an interior maximizer such that $\partial_\alpha\hat\ell(u,\hat\alpha(u))=0$ for each given $u \in {\cal U}$. Hence, there exists some $\bar \alpha(u) \in {\cal A}$, $u \in {\cal U}$, on the line segment connecting $\alpha_0(u)$ and $\hat\alpha(u)$ such that
\[
\sqrt{n}(\hat\alpha-\alpha_0)(u) = -[\partial^2_{\alpha\alpha}\hat\ell(u,\bar\alpha)]^{-1}\sqrt{n}\partial_\alpha \hat\ell(u,\alpha_0).
\]
By Lemma~\ref{lem:sec_delt}, $\sqrt{n}\partial_\alpha \hat\ell(u,\alpha_0) \Rightarrow \tau(u)\varphi_u(\mathbb{C}(\cdot,1))$ in $\ell^\infty({\cal U})$.   Recall from Eqs.~\eqref{eq:restricted} and \eqref{eq:ellinf} the feasible and infeasible likelihood contributions, respectively. We will show that uniformly in $u \in {\cal U}$
\[
\partial^2_{\alpha\alpha}\hat\ell(u,\bar\alpha) \rightarrow_p \Ex[\partial^2_{\alpha\alpha}\ell_t(u,\alpha_0)] = -\tau^2(u)\sum_{j=1}^4\frac1{p_j(u)}.
\]
We will do so in two steps: First, we show
\begin{align}\label{eq:alphaproofdiff}
    \ssup\limits_{(u,\alpha) \in {\cal U}  \times {\cal A}}\Big|\frac1{n}\sum_{t=1}^n (\partial_{\alpha\alpha}^2\hat\ell_{t}(u,\alpha)- \partial_{\alpha\alpha}^2\ell_{t}(u,\alpha))\Big| = o_p(1)
\end{align}
to then show
\begin{align}\label{eq:alphaulln}
      \ssup\limits_{(u,\alpha) \in {\cal U}  \times {\cal A}}\Big|\frac1{n}\sum_{t=1}^n (\partial_{\alpha\alpha}^2\ell_{t}(u,\alpha)- \Ex[\partial_{\alpha\alpha}^2\ell_{t}(u,\alpha)])\Big| = o_p(1).
    \end{align}
The claim then follows from the Proposition~\ref{prop:alpha} in conjunction with the continuity of $(u,\alpha) \mapsto \Ex[\partial_{\alpha\alpha}^2\ell_{t}(u,\alpha)]$ on the compact set ${\cal U}  \times {\cal A}$. Now, $\partial_{\alpha\alpha}^2\hat \ell_{t}(u,\alpha) = \sum_{j=1}^4\hat d_{j,t}(u)^\top\partial_{\alpha\alpha}^2\log p_j(u,\alpha)$ and $\partial_{\alpha\alpha}^2 \ell_{t}(u,\alpha) = \sum_{j=1}^4d_{j,t}(u)^\top\partial_{\alpha\alpha}^2\log p_j(u,\alpha)$. Applying the chain-rule and using $\varrho(u,\alpha) = {\sf tanh}(\alpha(u))$ we get
\begin{align*}
\partial_{\alpha\alpha}^2 \llog p_j(u,\alpha)
= \,&
(1-\varrho(u,\alpha)^2)^2
\bigg[
\frac{e_j\, C_{\varrho\varrho}(u,\alpha)}{p_j(u,\alpha)}\\
\,& -
\frac{C_{\varrho}(u,\alpha)^2}{p_j(u,\alpha)^2}
\bigg]
-
2\varrho(u,\alpha)(1-\varrho(u,\alpha)^2)\,
\frac{e_j\, C_{\varrho}(u,\alpha)}{p_j(u,\alpha)},
\end{align*}
where  $e \coloneqq (1,-1,-1,1)^\top$ and
$$
C_{\varrho}(u,\alpha) = \phi(\Phi^{-1}(u_1),\Phi^{-1}(u_2);\varrho(u,\alpha)) 
$$
and, with $z_i \coloneqq  \Phi^{-1}(u_i)$
$$
C_{\varrho\varrho}(u,\alpha) = C_{\varrho}(u,\alpha)\frac{z_1z_2+(1-z_1^2-z_2^2)\varrho(u,\alpha)+z_1z_2\varrho^2(u,\alpha)-\varrho^3(u,\alpha)}{(1-\varrho^2(u,\alpha))^2}.
$$
Because $\Ex[d_{j,t}(u)] = p_j(u,\alpha_0)=p_j(u)$, and, $\varrho(u,\alpha) = \rho(u)$ at $\alpha = \alpha_0$ it follows that $\Ex[\partial^2_{\alpha\alpha}\ell_t(u,\alpha_0)] = -\tau^2(u)\sum_{j=1}^4\frac1{p_j(u)}$. Now, it will be shown that there exists a constant $c \in (0,\infty)$ such that
\begin{align}\label{eq:supprob}
\mmax\limits_{1\leq j \leq 4} \ssup\limits_{u \in {\cal U}}\ssup\limits_{\alpha \in {\cal A}}\partial_{\alpha\alpha}^2\llog p_j(u,\alpha) \leq  c.  
\end{align}
To see this, note that there is a constant $c$ such that  $\mmax\limits_{1\leq j \leq 4}\iinf\limits_{u \in {\cal U}}\iinf\limits_{\alpha \in {\cal A}}p_j(u,\alpha) \geq c > 0$. Moreover, for each $j$, $\partial_\varrho p_j(u,\alpha) = e_j C_\varrho(u,\alpha)$ and $\partial^2_{\varrho\varrho} p_j(u,\alpha) = e_j C_{\varrho\varrho}(u,\alpha)$, while $\partial_\alpha \varrho(u,\alpha)= 1-\varrho^2(u,\alpha)$ and $\partial^2_{\alpha\alpha} \varrho(u,\alpha)= -2\varrho(u,\alpha)(1-\varrho^2(u,\alpha))$.  Moreover,
$$
\mmax\{\ssup\limits_{u \in {\cal U}}\ssup\limits_{\alpha \in {\cal A}}|C_{\varrho}(u,\alpha)|,\ssup\limits_{u \in {\cal U}}\ssup\limits_{\alpha \in {\cal A}}|C_{\varrho\varrho}(u,\alpha)|\}\leq c,
$$
which shows Eq.~\eqref{eq:supprob}.  By Eq.~\eqref{eq:dhatdiff}, $\ssup_{u}|\frac1{n}\sum_{t=1}^n \hat d_{j,t}(u)-\frac1{n}\sum_{t=1}^n d_{j,t}(u)| = O_p(n^{-1/2}).$ Combining the preceding display with Eq.~\eqref{eq:supprob} yields \eqref{eq:alphaproofdiff}. Eq.~\eqref{eq:alphaulln} follows uppon applying  \citet[Lemma 2.4]{newmc:94} because: ($i$) $\partial_{\alpha\alpha}^2\ell_t(u,\alpha)$ is a sequence of {\sf IID} random variables, ($ii$) $(u,\alpha) \mapsto \partial_{\alpha\alpha}^2\ell_t(u,\alpha)$ is continuous, ($iii$) $\ssup_{u \in {\cal U}}\ssup_{\alpha \in {\cal A}}|\partial_{\alpha\alpha}^2\ell_t(u,\alpha)| \leq c$,  ($iv$) and ${\cal U} \times {\cal A}$ is compact. This completes the proof. \hfill $\square$

\section{Proof of Proposition 2}

By the mean-value theorem, for each $(u,s) \in {\cal U} \times [0,1]$, there exists $\bar\alpha(u,s)$ (random) on the line segment between $\hat\alpha(u)$ and $\alpha_0(u)$ such that:
\begin{align*}
\sqrt{n}\nabla_\beta \hat\ell(u,s,\hat\theta) = \sqrt{n}\nabla_\beta \hat\ell(u,s,\theta_0)  + \sqrt{n}(\hat\alpha-\alpha_0)(u)  \hat J(u,s,\bar\alpha(u,s)),
\end{align*}
where  $\hat J(u,s,\alpha) \coloneqq \nabla^2_{\beta\alpha} \hat\ell(u,s,(\alpha,0)) = \partial_\alpha \nabla_\beta   \hat\ell(u,s,(\alpha,0))$, 
\begin{align*}
\hat J(u,s,\alpha)  = \tau(u,\alpha)\,\frac1{n}\sum_{t=1}^{\floor{sn}}  Z_t b(u,\alpha)^\top \hat d_t(u,s),
\end{align*}
recalling $\tau(u,\alpha) = C_\varrho(u,\alpha)(1-\varrho^2(u,\alpha))$ and $b(u,\alpha) \coloneqq (b_1(u,\alpha),\dots,b_4(u,\alpha))^\top$,
\[
b_j(u,\alpha) \coloneqq k(u,\alpha) \frac{e_j}{p_j(u,\alpha)}-\tau(u,\alpha) \frac{1}{p_j^2(u,\alpha)}, \quad j \in \{1,\dots,4\},
\]
with $e = (1,-1,-1,1)^\top$, $k(u,\alpha)\coloneqq  (1-\varrho^2(u,\alpha))\frac{C_{\varrho\varrho}(u,\alpha)}{C_{\varrho}(u,\alpha)}-2\varrho(u,\alpha).$ 
We first show,
\begin{align}
   \ssup\limits_{(u,s,\alpha) \in {\cal U} \times [0,1] \times {\cal A}} \|\frac1{n}\sum_{t=1}^{\floor{sn}}  Z_tb(u,\alpha)^\top (\hat d_{t}(u,s)-d_{t}(u))\| = o_p(1),
\end{align}
where we recall that $d_t(u)$ is based on the the true ranks.
To see this, note that there are $c_0,c_1 \in (0,\infty)$ such that $\mmax_j\ssup_{(u,\alpha) \in {\cal U}\times {\cal A}} |b_j(u,\alpha)| \leq c_0$ so that, by Lemma~\ref{lem:seqtilde}, 
\begin{align*}
   \ssup\limits_{(u,s,\alpha) \in {\cal U} \times [0,1] \times {\cal A}} \|\frac1{n}\sum_{t=1}^{\floor{sn}} & Z_tb(u,\alpha)^\top (\hat d_{t}(u,s)-d_{t}(u))\| \\
\,& \leq c_1 \ssup\limits_{(u,s) \in {\cal U} \times [0,1]}\|G_n(u,s,\hat\lambda)-\tilde G_n(u,s,\lambda_0)\| = o_p(1).
\end{align*}
Thus, uniformly in $(u,s,\alpha)\in {\cal U} \times [0,1] \times {\cal A}$:
$$
\hat J(u,s,\alpha) =  \tau(u,\alpha)\,\frac1{n}\sum_{t=1}^{\floor{sn}}  Z_t  b(u,\alpha)^\top d_{t}(u) + o_p(1).
$$
Because $\Ex[d_{j,t}(u)] = p_j(u,\alpha_0) = p_j(u)$ and $\sum_j e_j = 0$, $\Ex[w_{1,t}(u,\alpha)]=0$ at $\alpha = \alpha_0$, we get
\[
\Ex[Z_0b(u,\alpha_0)^\top d_{0}(u)] = \mu b(u,\alpha_0)^\top p(u) = -\mu \tau(u)\sum_{j=1}^4\frac{1}{p_j(u)},
\]
with $p(u) = (p_1(u),\dots,p_4(u))^\top$. Next, because $b(u,\alpha)$ are uniformly bounded in $(u,\alpha)$, there exists $c \in (0,\infty)$ such that
\begin{align*}
\ssup\limits_{(u,s,\alpha)\in {\cal U} \times [0,1] \times {\cal A}}\| \frac1{n}\sum_{t=1}^{\floor{sn}}&Z_tb(u,\alpha)^\top(d_t(u)-p(u))\| \\
& \leq \frac{c}{\sqrt n}\ssup\limits_{(u,s) \in [0,1]^2 \times [0,1]} |\tilde {\mathbb G}_n(u,s;\lambda_0)| = O_p(n^{-1/2}),
\end{align*}
where the order of magnitude follows again from Lemma~\ref{lem:seqtilde}.  Thus, uniformly in $(u,s,\alpha)\in {\cal U} \times [0,1] \times {\cal A}$:
\[
\hat J(u,s,\alpha) = s\mu\tau(u,\alpha)b(u,\alpha)^\top p(u)+o_p(1).
\]
Thus, because $(u,\alpha) \mapsto \tau(u,\alpha)b(u,\alpha)^\top p(u)$ is continuous on the compact set ${\cal U} \times {\cal A}$, Proposition~\ref{prop:alpha} yields
\[
\hat J(u,s,\bar\alpha(u,s)) = s\mu\tau(u)b(u,\alpha_0)^\top p(u)+o_p(1).
\]
Because, by Proposition~\ref{prop:alpha}, $\sqrt{n}(\hat\alpha-\alpha_0)(\cdot) \Rightarrow {\mathbb C}(\cdot,1)/\tau(\cdot)$ in in $\ell^\infty({\cal U})$, Slutzky's theorem and Lemma~\ref{lem:sec_delt} yield in $\ell^\infty({\cal U} \times [0,1]; \mathbb R^k)$:
\[
\sqrt{n}\nabla_\beta \hat\ell(u,s,\hat\theta) \Rightarrow \tau(u)\varphi_u(\mathbb{G}(\cdot,s))+s \mu  {\mathbb C}(u,1)   b(u,\alpha_0)^\top p(u),
\]
with  ${\mathbb C}(u,1)   b(u,\alpha_0)^\top p(u) = -\tau(u)\varphi_u(\mathbb{C}(\cdot,1))$. The claim follows. \hfill $\square$
\end{appendices}
\end{document}